\documentclass[journal]{IEEEtran}
\usepackage{amsmath,amssymb,amsfonts,graphicx,epsfig,cite,fancyhdr,amsthm,graphicx}
\usepackage{dsfont, color}
\usepackage[dvipsnames]{xcolor}
\usepackage{subfigure}
\usepackage[center]{caption}
\usepackage{chngpage}
\usepackage{multicol}
\usepackage{array}
\usepackage{algpseudocode}
\usepackage{multirow, rotating,mathrsfs }
\usepackage[square, comma, numbers, sort&compress]{natbib}
\usepackage{bigints}
\usepackage{epstopdf}
\usepackage{dsfont}
\usepackage{multicol, everyhook}

\theoremstyle{definition}
\newtheorem{definition}{Definition}

\newtheorem{lemma}{Lemma}

\newtheorem{result}{Result}
\newtheorem{model}{Model}

\usepackage{tikz,pgfplots}
\usetikzlibrary{patterns}
\usetikzlibrary{shapes,snakes}
\usetikzlibrary{arrows, calc}
\usetikzlibrary{plotmarks}

\makeatletter %new code
\pgfdeclarepatternformonly[\LineSpace]{my north east lines}{\pgfqpoint{-1pt}{-1pt}}{\pgfqpoint{\LineSpace}{\LineSpace}}{\pgfqpoint{\LineSpace}{\LineSpace}}%
{
    \pgfsetcolor{\tikz@pattern@color} %new code
    \pgfsetlinewidth{0.4pt}
    \pgfpathmoveto{\pgfqpoint{0pt}{0pt}}
    \pgfpathlineto{\pgfqpoint{\LineSpace + 0.1pt}{\LineSpace + 0.1pt}}
    \pgfusepath{stroke}
}
\makeatother %new code
\newdimen\LineSpace
\tikzset{
    line space/.code={\LineSpace=#1},
    line space=3pt
}

\title{Modeling and Analysis of Energy Efficiency and Interference for Cellular Relay Deployment}
\author{Fanny Parzysz, \IEEEmembership{Student Member, IEEE}, Mai Vu, \IEEEmembership{Senior Member, IEEE}, Fran\c cois Gagnon, \IEEEmembership{Senior Member, IEEE}%
\thanks{This work is supported in part by the FQRNT and the NSERC - Ultra Electronics Industrial Chair in Wireless Emergency and Tactical Communications. F. Parzysz and F. Gagnon are with \'Ecole de Technologie Sup\'erieure, Montreal, Canada. M. Vu is with Tufts University, Medford, USA. (Fanny.Parzysz@lacime.etsmtl.ca; Mai.Vu@tufts.edu; Francois.Gagnon@etsmtl.ca) }
}

\begin{document}
\maketitle
%\fancyhead{\small \textit{Submitted to IEEE JSAC Series on Green Communications and Networking and soon available on arXiv}}

%\newcounter{algorithmcounter}
%\setcounter{algorithmcounter}{\value{table}}
%\addtocounter{algorithmcounter}{1}
%
%\newcounter{tableau}
%\setcounter{tableau}{1}

%\vspace*{-5pt}
\begin{abstract}
% 75- to 200-word abstract
By relying on a wireless backhaul link, relay stations enhance the performance of cellular networks at low infrastructure cost, but at the same time, they can aggravate the interference issue. In this paper, we analyze for several relay coding schemes the maximum energy gain provided by a relay, taking into account the additional relay-generated interference to neighboring cells. 
First, we define spatial areas for relaying efficiency in log-normal shadowing environments and propose three easily-computable and tractable models. These models  allow the prediction of 1) the probability of energy-efficient relaying, 2) the spatial distribution of energy consumption within a cell and 3) the average interference generated by relays.
Second, we define a new performance metric that jointly captures both aspects of energy and interference, and characterize the optimal number and location of relays. These results are obtainable with significantly lower complexity and execution time when applying the proposed models as compared to system simulations. We highlight that energy-efficient relay deployment does not necessarily lead to interference reduction and conversely, an interference-aware deployment is suboptimal in the energy consumption. We then propose a map showing the optimal utilization of relay coding schemes across a cell. This map combines two-hop relaying and energy-optimized partial decode-forward as a function of their respective circuitry consumption. Such a combination not only alleviates the interference issue, but also leads to a reduction in the number of relays required for the same performance.

\end{abstract}

\bstctlcite{IEEEexample:BSTcontrol}


\section{Introduction}
\label{introduction}

In the urge to limit the energy consumption of cellular networks while maintaining service quality and ubiquitous access, relaying is a flexible and economical solution to enhance performance, eliminate coverage dead zones or alleviate traffic hot zones. It is envisioned as part of next-generation cellular networks, along with pico- and femtocells \cite{aliu2013}. 
Unlike small cells, relay stations are not connected to the core network through a wireline backhaul connection but have to rely on wireless transmission to access the base station. 
This offers significant infrastructure cost reduction and deployment flexibility but, at the same time, can aggravate the interference issue. Exploring energy-optimized relaying jointly with interference reduction and choice of coding scheme opens new perspectives for efficient relay deployment.

%%%%%%%%%%%%%%%%%%%

\subsection{Motivation and Prior work}

%various performance metrics have been considered, such as capacity enhancement \cite{dziong2012,Zolotukhin2012, Minelli2014}, coverage extension \cite{joshi2011, khakurel2012} and energy minimization \cite{b1,b2}.

As highlighted in \cite{book_green_network_ch6}, the relay location within a cell can significantly affect the system performance. Its optimization is necessary to guarantee maximal gains within reasonable deployment cost and to avoid poor relay locations with negligible gains. Substantial efforts have been paid in optimizing the relay location with regards to capacity  \cite{dziong2012, Zolotukhin2012, Minelli2014}, coverage \cite{joshi2011, khakurel2012} and energy \cite{Guiying2012, Chandwani2010}.
The serving area of relay stations has generally been envisaged as a small circular zone around each relay \cite{Peng2015} or covers the cell edge exclusively as in \cite{Chandwani2010}. However, such serving areas are neither energy- nor capacity-optimized. An analysis of the serving area has been provided in \cite{Minelli2014conf} for capacity and in  \cite{Journal2} for energy, via the newly-defined relay efficiency area.
These works, however, account only for the transmit energy but not circuitry consumption, and consider only path-loss but not shadowing, which is a significant cause of signal degradation.
% that optimizing the cell area for relay utilization can be much  the service area of relay station from an energy and coverage perspective , 

%To investigate the relay efficiency, the concept of Relay Efficiency Area (REA) has been introduced in \cite{Journal2}. Both coverage extension and energy gains are included in this concept, leading to a cell area for relay utilization, denoted much wider than in other geometrical patterns proposed in the literature, e.g. \cite{Chandwani2010}. In the following, we refine this concept of REA to include the effect of shadowing. 
%
%extend the notion of serving area: femto-cell viw as in ... coverage view as in ...
%work of Minelli
%but...

On the one side, analyzing only the useful transmit power, i.e. the power radiated by the antenna, allows fair characterization of the network upper-bound performance, as done in \cite{joshi2011, khakurel2012, Zolotukhin2012,  dziong2012, Minelli2014}.
However, it is generally not sufficient for an energy-efficient analysis. Depending on the considered technology and hardware quality, the energy dissipated in circuitry for site cooling, network maintenance and signal processing (even if no transmission is operated), may dominate the overall consumption \cite{correia2010,book_green_network} such that the performance limits predicted by theory may not be realized. This is particularly true when the network does not operate at full load.% and this energy loss should be included.

On the other side, considering the transmit energy alone provides a basis for complementary analysis of multi-cell networks, as it is directly related to inter-cell interference (ICI), another great challenge of cellular networks. 
Unlike base stations, relay stations are generally equipped with omni-directional antennas and can drastically increase the interference to a neighboring cell. Several questions remain open: Does energy-efficient relay deployment necessarily lead to interference reduction? Does deploying few relay stations far from the cell edge but serving a large part of the cell generate less interference than deploying many relay stations transmitting at low power but potentially closer to cell edge? Answering these questions necessitates the understanding of whether there exists a trade-off between energy-efficient and interference-aware relay deployment in cellular networks.

The ICI constraint has been investigated in the context of the relay placement for capacity and coverage enhancement in \cite{joshi2011, khakurel2012, Hamdi2012, Minelli2014}, in which several models have been proposed. In \cite{joshi2011, khakurel2012, Hamdi2012}, 
relays are located according to a predefined pattern, e.g. around a circle centered at their serving base station. In this case, interference appears to be very pessimistic, since each neighboring node (among the base stations and relays of each cell) is assumed to interfere with the reference user in each time slot with full transmit power.
An existing model for interference evaluation allows for interference from nearby stations to be computed in closed-form, while afar interferers are modeled as a continuum rather than a discrete set \cite{Minelli2014}. 
While such models are particularly suitable for capacity maximization or coverage extension,
they are not useful in analyzing or optimizing energy consumption because of the assumption that all stations transmit at the maximum power while the actual transmitted power can be much less. Understanding how energy consumption is distributed across the cell is essential for an energy approach and, to the best of our knowledge, no interference analysis has been performed for energy-optimized relay-assisted networks.

\subsection{Main contributions and Paper overview}

Our objective is to investigate how the deployment of decode-forward-based (DF) relay stations, specifically their number and location, can reduce the overall energy consumption of a cell and how it affects at the same time the performance of neighboring cells, due to the additional interference. 
First, we propose spatial definitions of relaying efficiency, namely the Efficiency Areas, as the cell area for which a given performance requirement is satisfied.
For shadowing environments, we define:
\begin{itemize}
\item the Relay Efficiency Area (REA), inside which a user has at least a probability $\mathbb{P}_T$ to be served by the relay; this REA extends the model in \cite{Journal2} to include shadowing.
\item the Energy Efficiency Area (EEA), for which the energy consumption does not exceed $\mathsf{E}_T$;
%\item the Interference Efficiency Area (IEA), for which the relay station does not generate more interference than a given level $\mathsf{I_{IEA}}$.
\end{itemize}
For the energy consumption, we account for (a) the transmit energy, (b) the additional energy consumed by the relay station for signal processing (decoding and re-encoding), (c) the energy loss at the RF amplifier and (d) the transmission-independent energy offset dissipated for network maintenance and site cooling, also called idle energy.

Simulating the network performance offers wide possibilities but is time-consuming and can rapidly turns out to be unfeasible in shadowing environments.
Moreover, they provide neither generalization to other simulation settings nor performance limits like an analysis based on capacity bounds. Therefore, we propose in this paper easily-computable models for relaying probability (REA) and energy consumption (EEA), which allow meaningful performance analysis without requiring extensive simulations. Such models have wide application and offer valuable support for load balancing, resource management or base-station switch off.

Second, based on the proposed models for REA and EEA, we define a new framework for interference analysis. We characterize the average relay-generated interference to a given neighboring user and propose a model for it. Next, we define a new performance metric, called $\Gamma$, which balances the energy gain provided by relays and the additional interference they generate.% to neighboring cells.

Third, we apply the proposed models to investigate efficient relay deployment. With an energy-efficiency perspective, we characterize how the circuitry energy consumption, dissipated for signal processing and network maintenance, affects the network performance, as a function of the cell radius and the number of relays.
%We show that deploying many relay stations transmitting at low power but potentially closer to cell edge appears more efficient than deploying few relay stations serving a large part of the cell but far from cell edge.
We also highlight that energy-efficient relay deployment does not necessarily lead to interference reduction and conversely, an interference-aware relay deployment is suboptimal in terms of energy reduction.

Finally, most approaches to efficient relaying are limited to the analysis of the simple two-hop relaying scheme. This scheme provides some performance enhancement but fails to capture the true potential of relays.
Thus, in addition to two-hop relaying, we consider two energy-optimized partial decode-forward schemes of \cite{Journal1}, which respectively minimize the network consumption and the relay consumption only. The second scheme is particularly relevant for analyzing relay-generated interference.
Given the increased circuitry consumption of such optimized schemes, we draw a map showing the cell areas for the optimal use of each scheme that maximizes the energy-to-interference ratio $\Gamma$. The resulting map not only alleviates the interference issue, but also allows the reduction of the number of relays necessary to reach the same performance.

%%%%%%%%%%%%%%%%%%%%%%%%%%%%
%\subsubsection{Paper overview}

This paper is organized as follows. The cell configuration and channel model are described in Section \ref{sec:system_model}. % We also characterize the scenarios considered for performance evaluation. 
The reference coding schemes  and related model for energy consumption are described in Section \ref{sec:model_energy_consumption}.
The models for REA and EEA are proposed and analyzed respectively in Sections \ref{sec:REA} and \ref{sec:EEA}. The new framework for interference analysis is characterized in Section \ref{sec:IEA}. The application of these model in energy- and interference-aware relay deployment is covered in Section \ref{sec:performance} and Section \ref{sec:conclusion} concludes this paper.

\section{System Model for a Relay-Aided Cell}
\label{sec:system_model}

We present in this section the main assumptions on the cell environment and channel model. We use the following notation:
%\textit{Notation:} %This paper uses the following notation. 
BS stands for the base station, RS for the relay station and U for the user. Calligraphic $\mathcal{R}$ refers to the user rate. Upper-case letters denote constant distances, such as $\mathsf{D}$, $\mathsf{R}$ for radius, or $\mathsf{H}$ for height. Lower-case letters stand for variable distances or angles, e.g. the mobile user coordinates $(x,y)$ or $(r,\theta)$. Subscripts $_{d}$, $_{b}$ and $_{r}$ respectively refer to the link from user to BS (direct link), from relay to BS and from user to RS. $\mathbb{P}$ is used for probabilities and $\mathbb{E}$ for expectation. Finally, DTx refers to direct transmission, while RTx stands for relay-aided transmission, in the considered coding scheme.

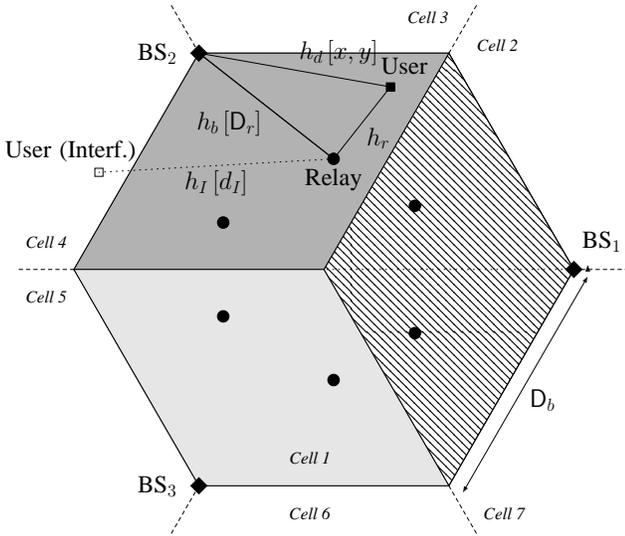
\begin{figure} 
\centering \centering \resizebox{0.95\columnwidth}{!}{%

\begin{tikzpicture}[scale=4,cap=round,>=latex]
% Radius of regular polygons
  \newdimen\R
  \newdimen\r 
  \R=1.8cm
  \r=0.8cm
  \coordinate (center) at (0,0);
  
  \coordinate (bs1) at (\R,0);
  \coordinate (bs2) at (120:\R);
  \coordinate (bs3) at (240:\R);
  
  \coordinate (rs1) at (35:\r);
  \coordinate (rs11) at (-35:\r);
  
  \coordinate (rs22) at (155:\r);
  \coordinate (rs2) at (85:\r);
  
  \coordinate (rs3) at (275:\r);
  \coordinate (rs33) at (205:\r);

  \coordinate (user) at (70:1.4cm);
  \coordinate (poor_user) at (-0.9*\R,0.7cm);

 \draw (0:\R)
     \foreach \x in {60,120,...,360} {  -- (\x:\R) };
%              -- cycle (300:\R) node[below] {$\csc \theta$}
%              -- cycle (240:\R) node[below] {$\sec \theta$}
%              -- cycle (180:\R) node[left] {$\tan \theta$}
%              -- cycle (120:\R) node[above] {$\sin \theta$}
%              -- cycle (60:\R) node[above] {$\cos \theta$}
%              -- cycle (0:\R) node[right] {$\cot \theta$};

  \draw { (0:\R) -- (60:\R) -- (center) -- (300:\R) --(0:\R) } [fill=white!50!black, pattern=north west lines];
  \draw { (120:\R) -- (60:\R) -- (center) -- (180:\R) --(120:\R) } [fill=white!70!black];
  \draw { (300:\R) -- (center) -- (180:\R) -- (240:\R) --(300:\R) } [fill=white!90!black];
  
	\draw (bs1) node[diamond, draw, fill = black] {};
	\draw (bs2) node[diamond, draw, fill = black] {};
	\draw (bs3) node[diamond, draw, fill = black] {};
	
	\draw (rs1) node[circle, draw, fill = black] {};
	\draw (rs2) node[circle, draw, fill = black] {};
	\draw (rs3) node[circle, draw, fill = black] {};
	\draw (rs11) node[circle, draw, fill = black] {};
	\draw (rs22) node[circle, draw, fill = black] {};
	\draw (rs33) node[circle, draw, fill = black] {};
	
	\draw (user) node[rectangle, draw, fill= black] {};
	\draw (poor_user) node[rectangle, draw, fill= white] {};
	
	\draw { (bs2) -> (rs2)};
	\draw[loosely  dotted, thick] { (bs1) -- (center)};
	
%	\draw[loosely  dotted,thick] (bs2) edge node [above left] {\huge $h_{bi}$} (poor_user);
	\draw[loosely  dotted,thick] (rs2) edge node [below] {\huge $h_{I} \left[d_I \right]$} (poor_user);	
	
  \draw[thick] (bs2) edge node [below left] {\huge $h_b \left[\mathsf{D}_r \right]$} (rs2);
  \draw[thick] (user) edge node [below right] {\huge $h_r $} (rs2);
  \draw[thick] (bs2) edge node [above right] {\huge $h_d \left[x,y \right]$} (user);
  %\draw ($(bs2)+(0.25,-0.20)$) edge [bend right] node [right] {\huge $\theta$} ($(bs2)+(0.3,-0.05)$);
  
  \draw ($(bs1)+(0.2,0.2)$) node {\huge BS$_1$};
  \draw ($(bs2)+(-0.3,0)$) node {\huge BS$_2$};
  \draw ($(bs3)+(-0.3,0)$) node {\huge BS$_3$};
  \draw ($(rs2)+(0,-0.15)$) node {\huge Relay};
  \draw ($(user)+(0.1,0.15)$) node {\huge User};
  \draw ($(poor_user)+(0.3,0.15)$) node [anchor=east] {\huge User (Interf.)};
  
  \draw[<->,thick] ($(bs1)+(0.1,-0.05)$) edge node [below right] {\huge $\mathsf{D}_b $} ($(300:\R)+(0.1,-0.05)$);
  
  \draw[dashed] { (0:\R) -- (0:2.2)};
  \draw[dashed] { (60:\R) -- (60:2.2)};
  \draw[dashed] { (120:\R) -- (120:2.2)};
  \draw[dashed] { (180:\R) -- (180:2.2)};
  \draw[dashed] { (240:\R) -- (240:2.2)};
  \draw[dashed] { (300:\R) -- (300:2.2)};
  
 \draw ($(bs3)+(0.8,0.2)$) node {\Large \textit{Cell 1}};
 \draw ($(60:\R)+(60:0.3)+(0.2,-0.2)$) node {\Large \textit{Cell 2}};
 \draw ($(60:\R)+(60:0.3)+(-0.3,0)$) node {\Large \textit{Cell 3}};
 \draw ($(180:\R)+(-0.2,0.2)$) node {\Large \textit{Cell 4}};
 \draw ($(180:\R)+(-0.2,-0.2)$) node {\Large \textit{Cell 5}};
 \draw ($(bs3)+(0.8,-0.2)$) node {\Large \textit{Cell 6}};
 \draw ($(300:\R)+(0.4,-0.2)$) node {\Large \textit{Cell 7}};

\end{tikzpicture}
}
\caption{System model for a hexagonal cell aided by 6 relays}
\label{fig:cell_system}
\end{figure}

\subsection{Cell topology}

We consider an hexagonal cell with edge distance $\mathsf{D}_b$, as depicted in Figure \ref{fig:cell_system}. It consists of 3 sectorized base stations BS$_i$, $i \in \left\lbrace 1,3 \right\rbrace$, located above surrounding buildings. We assume a typical radiation pattern for each base station and consider the antenna gain as given in \cite{kathrein}.
Each 120$^\circ$-sector $i$ is served by $N_r$ relay stations, all equipped with omni-directional antennas. 
%, denoted RS$_j^{(i)}$, $j \in \left\lbrace 1,N_r \right\rbrace$.
%, with height $\mathsf{H}_r$ varying from below to above rooftop.
A given relay is at a distance of $\mathsf{D}_r$ from its assigned BS.  %Unlike BS, relays are equipped with omni-directional antennas, implying that the ICI generated by RTx can be significantly higher than for DTx. 
Finally, we assume a mobile user positioned at $(x,y)$. 
%We denote $r_d$ the distance from the user to the base station.
%, and $\theta$, the related angle. Moreover, we assume that the user
It is associated with the closest base station and the relay that provides the highest energy gain. 
We do not consider techniques for relay selection or multi-relays schemes but the proposed analysis can be adapted to such configurations.
We define the maximal coverage by the outage requirement $\mathbb{P}_\text{out}$ that has to be maintained throughout the whole cell.

\subsection{Description of the relay channel}

%\subsubsection{Half-duplex relaying}

We consider  half-duplex relaying performed, without loss of generality, in time division. %Nevertheless, the following analysis is valid with other multiplexing schemes.
We assume that the relay operates on the same frequency resource as the user it serves (namely \textit{in-band} relaying in LTE-systems).
The multiple access strategy allows orthogonality between users within the cell, such that only one user is served for a given time and frequency resource.
%ATTENTION: for Full-DF / two-hop relaying => out-band relaying is also valid. Since we do not consider, as baseline model, interference.
A downlink transmission is carried out in two phases of equal duration. 
%We denote the transmitted codewords $X_{1}$ and $X_{2}$ for the base station at each phase respectively, and $X_r$ for the relay. $Y_r$ and $Y_1$ respectively are the received signals at the relay and user at the end of phase 1, $Y_2$ is the received signal at the user at the end of phase 2. The half-duplex relay channel is written as:
%\begin{equation}
%\begin{split}
%Y_1 & = h_d X_1 + Z_1 \; ; \quad \quad
%Y_r = h_b X_1 + Z_r  \; ; \quad \quad
%Y_2 = h_d X_2 + h_r X_r + Z_2 
%\end{split}
%\label{eq:channel}
%\end{equation}
%where $Z_r$, $Z_1$ and $Z_2$ are independent additive white Gaussian Noises (AWGN) with equal variance $N$.
We assume Gaussian channels with independent additive white Gaussian Noises (AWGN) with equal variance $N$ on all links.
We respectively denote $h_d$, $h_b$ and $h_r$ the channel coefficients from base station to user (direct link), from base station to relay (wireless backhaul link) and from relay to user. 
In addition, $h_{I}$ refers to the channel from the interfering relay to a user in a neighboring cell.

%(resp. $h_{bi}$) (resp. the base station)
%One can point out that inter-cell interference is omitted in this analysis. 
%Since interference-mitigation techniques, such as Fractional Frequency Reuse (FFR) or base station coordination, have been spurred for next-generation cellular standards, we consider as a baseline a cell with reasonably low inter-cell interference, leaving this topic to future extension. 

%\subsubsection{Path-loss and shadowing}

We assume that the transmitted signal is degraded by both path-loss and shadowing.
Note that we do not consider small-scale fading, since fading coefficients are rarely known for power allocation and terminals are designed to be sufficiently robust against such small-scale parameters, especially for relatively low mobility environments.
Hence, this work is based on a per-block capacity analysis and can be understood as the mean relaying probability and energy consumption over a sufficient period of time such that small-scale fading is averaged out.

 %As example, such analysis has been provided for femtocells networks in \cite{Dhillon2012}.

The BS, RS and user have different heights, moving neighborhoods and transmission ranges, such that different links have different properties, particularly in terms of path-loss. 
To fit this heterogeneity, we consider the channel model proposed in the WINNER II project \cite{winner}, where the path-loss of link $k \in \left\lbrace d,b,r, I\right\rbrace$, denoted $\gamma_{k}$, is specified by four parameters $A_{k}$, $B_{k}$, $C_{k}$ and $D_{k}$ depending on the global location of the transmitter and receiver (street level, rooftop...).
The shadowing coefficient $s_k$ is modelled as a log-normal random variable of variance $\sigma_k^2$.  We assume that all $s_k$'s are mutually independent.
Channel gains are given by:
\begin{align}
\label{eq:pathloss}
& \vert h_{k} \vert^2 = \frac{s_k}{\gamma_{k}}= \frac{s_k}{K_{k} d^{\alpha_{k}}} \\
\text{with} \quad &
\left\lbrace \begin{array}{rl}
\alpha_k = A_k / 10 & \\
10 \log_{10} \left( K_k \right) & =  B_{k} + C_{k} \log_{10} \left(\frac{f_c}{5} \right) \\
& + D_{k} \log_{10}\left((\mathsf{H}_\text{Tx}-1)(\mathsf{H}_\text{Rx}-1)\right)
\end{array} 
\right. \nonumber
\end{align}
where $\mathsf{H}_\text{Tx}$ and $\mathsf{H}_\text{Rx}$ are the respective heights (in meters) of transmitter and receiver, $d$ is the distance (in meters) between them and $f_c$ the carrier frequency (in GHz). 
%Table \ref{tab:parameter} describes the path-loss model as well as the simulation parameters.

%
%One can point out that inter-cell interference is omitted in this analysis. 
%Here, aim is not on proposing a new power allocation (easily expressed when the amount of interference is know), but rather to...
%interference generated as a function of the relay location

\section{Coding schemes and Models \\ for Energy Consumption} 
\label{sec:model_energy_consumption}

We discuss the relaying schemes considered for analysis and their respective overall energy consumption. To simplify notation and focus on energy, we consider \emph{normalized transmissions of unitary length}, thus setting up a direct relation between energy and power.

\subsection{On the overall energy consumption}
\label{sec:overhead_energy}

The time-varying transmit energy, i.e. the energy radiated at the output of the RF antenna front-end, is referred as $E^{\mathsf{(RF)}}$. We assume individual energy constraints over the two transmission phases for the BS and RS, denoted as $\mathsf{E_B^{(m)}}$ and $\mathsf{E_R^{(m)}}$ respectively (with $^{\mathsf{(m)}}$ for maximum).
Most literature only accounts for this useful transmit energy, which is a fair assumption for capacity or coverage analysis, %, where the transmit energy dominates the overall consumption.
However, it is generally not sufficient for an energy analysis. % and the circuitry consumption should be considered as well.

First, significant energy is dissipated in circuit electronics for data transmission, especially in the RF amplifier which usually performs with considerably low efficiency. 
%To account for this inefficiency,
We consider the simplified, yet meaningful, approach of \cite{andreev2012}, where the amplifier inefficiency is assumed linear in the transmit energy $E^{\mathsf{(RF)}}$ and is characterized by a multiplicative coefficient, denoted $\eta_R$ for the relay and $\eta_B$ for the base station.

Second, we account for the circuitry consumption related to a transmission, i.e. to signal processing at the encoder and decoder. As in \cite{book_green_network_ch6}, we model this consumption by the following energy offsets: $\mathsf{E_B^{(Tx)}}$ for the BS transmission, $\mathsf{E_R^{(Rx)}}$ and $\mathsf{E_R^{(Tx)}}$ for the RS reception and transmission, and $\mathsf{E_U^{(Rx)}}$ for the reception at the mobile user.
Such energy offsets mostly depend on the quality of electronics and on the complexity of the signal processing performed at terminals, notably given by the considered coding scheme. 
Using simple relaying schemes can help decrease this energy consumption, but at the cost of potentially degraded network performance. %Thus, the choice of the coding scheme should account for such compromise.

Third, in addition to the energy related to the transmission itself, we consider the transmission-independent consumption, also referred as the idle energy. This is the energy dissipated for site cooling, network maintenance and additional signaling, which is consumed whether or not data is transmitted. As proposed in \cite{andreev2012}, it is modelled by an offset, consumed at each station and denoted as $\mathsf{E_R^{(idle)}}$ for relays, and $\mathsf{E_B^{(idle)}}$ for base stations.

\subsection{Coding schemes considered for analysis}
\label{sec:coding_scheme}

We now present the reference coding schemes and express for each scheme the overall energy consumption. In this work, we consider Gaussian signaling which accurately approximates OFDM-based communications \cite{OFDM_Gaussian}, as used in practical systems including LTE and WiMAX.
We consider downlink transmissions and assume that channel coefficients are known at transmitters to achieve the best possible performance. The power allocation aims at minimizing the energy consumption, given a fixed user rate $\mathcal{R}$.

\subsubsection{Direct transmission (DTx)}
Data is transmitted directly from the BS to the user over the two transmission phases, and not just only during the first one as done in most literature. This two-phase transmission allows fair performance comparison since, in this case, both direct and relay-aided transmissions have the same delay and consume the same time resource.
The energy consumption is given by:
\begin{align}
\hspace*{-12pt}\mathsf{E_{DTx}} = & \eta_B E_B^{\mathsf{(RF)}}
+ \left( \mathsf{E_B^{(Tx)}} + \mathsf{E_U^{(Rx)}} \right)
+ \left( \mathsf{E_B^{(idle)}} + N_r \mathsf{E_R^{(idle)}} \right)
\label{eq:overall_energy_DTx}
\end{align}
Setting $N_r=0$ gives the consumption of a reference scenario, where no relay station is deployed. However, when $N_r>0$, the energy consumption of DTx should account for the idle energy $N_r \mathsf{E_R^{(idle)}}$ dissipated at relay stations, even if DTx does not actually use those relays. 

\subsubsection{Two-hop relaying (2Hop)}
\label{sec:two-hop}

%In the literature, two main decode-forward schemes are generally considered: two-hop relaying, as  widely considered for practical systems and repetition-coded full decode-forward, as defined in \cite{laneman2004}. We will consider for analysis two-hop relaying only. However, we will show in Section \ref{sec:validation} that the proposed models are valid for other decode-forward schemes as well, including repetition-coded full decode-forward.

%In two-hop relaying,
Two-hop relaying is the simplest decode-forward scheme and thus, gives lower-bounds of the performance that can be achieved with DF-based relaying.
In this scheme, the user sends its message to the relay station with rate $2\mathcal{R}$ during Phase 1. Then, the relay decodes the message, re-encodes it with rate $2\mathcal{R}$ and forwards it to the destination during Phase 2. The base station finally decodes using only the signal received from the relay station (the direct link is ignored).
The energy consumption is expressed as:
\begin{align}
\mathsf{E_{2Hop}} = & \eta_B E_B^{\mathsf{(RF)}} +
\left( \mathsf{E_B^{(Tx)}} + \mathsf{E_R^{(Rx)}} \right)
+ \eta_R E_R^{\mathsf{(RF)}}
\nonumber \\&
 + \left( \mathsf{E_R^{(Tx)}} + \mathsf{E_U^{(Rx)}} \right)
+ \left( \mathsf{E_B^{(idle)}} + N_r \mathsf{E_R^{(idle)}} \right).
\end{align}
We denote $\mathsf{E_{2Hop}^{(dsp)}} = \mathsf{E_R^{(Rx)}} + \mathsf{E_T^{(Rx)}}$ the additional energy dissipated at the RS for decoding and re-encoding, in comparison with the energy offset of DTx described in Eq. \eqref{eq:overall_energy_DTx}.

\subsubsection{Optimized partial decode-forward schemes}

Finally, we consider for simulations the partial decode-forward (PDF) schemes optimized for energy proposed in \cite{Journal1}.
In this, only part of the initial message is relayed, the rest being sent via the direct link.
The base station divides the data message into two parts $m_r$ and $m_d$ using rate splitting. 
During the first transmission phase, $m_r$ is broadcast to both the relay station and the user. At the end of this phase, only the relay decodes $m_r$ and then re-encodes it. During the second phase, the relay sends $\tilde{m_r}$ and the base station jointly sends $(m_r,m_d)$ using superposition coding. At the end of phase 2, the user jointly decodes $m_r$ and $m_d$ to recover the initial message.
Several power allocations have been proposed in \cite{Journal1}. In this work, we consider the two following schemes:
\begin{itemize}
\item Energy-optimized PDF (EO-PDF), which minimizes the total transmit energy consumption and is both energy- and rate-optimal (referred to as G-EE in \cite{Journal1}).
\item Interference-at-Relay PDF (IR-PDF), which minimizes the energy transmitted by the relay only, thus minimizing the relay-generated interference (namely R-EE in \cite{Journal1}).
\end{itemize}
We refer the reader to \cite{Journal1} for the detailed power and rate allocation for $m_r$ and $m_d$. Such optimized schemes arguably require complex implementation and fine synchronization but they provide theoretical upper-bounds of the performance achievable with DF-based relaying.

Little information is available on the circuitry energy consumed by such coding techniques. To account for the increased complexity of these schemes, we consider an additional energy offset $\mathsf{E_{pdf}^{(dsp+)}}$ to the overall consumption as follows:
\begin{align}
\hspace{-10pt}\mathsf{E_{EO}}& = \eta_B E_B^{\mathsf{(RF)}} +
\left( \mathsf{E_B^{(Tx)}} + \mathsf{E_R^{(Rx)}} \right)
+ \eta_R E_R^{\mathsf{(RF)}}
\nonumber \\&
 + \left( \mathsf{E_R^{(Tx)}} + \mathsf{E_U^{(Rx)}} \right)
+ \left( \mathsf{E_B^{(idle)}} + N_r \mathsf{E_R^{(idle)}} \right)
+ \mathsf{E_{pdf}^{(dsp+)}}.
\end{align}
In the performance analysis of Section \ref{sec:performance}, we will consider several values for $\mathsf{E_{pdf}^{(dsp+)}}$.
As for two-hop relaying, we denote $\mathsf{E_{pdf}^{(dsp)}} = \mathsf{E_R^{(Rx)}} + \mathsf{E_T^{(Rx)}} + \mathsf{E_{pdf}^{(dsp+)}}$, representing the additional energy dissipated in circuitry to perform this relaying scheme.

%Hence, we can consider that the performance of any DF-based relaying scheme lies between these optimized schemes and the two-hop relaying scheme described above.

%%%%%%%%%%%%%%%%%%%%%%%%%%%%%%%%%%%%%%%%%%%%%%%%%%%%%%%
\section{Relaying Probability and Relay Efficiency Area with shadowing}
\label{sec:REA}

We first characterize the relaying efficiency in shadowing environments and define the Relay Efficiency Area (REA) as the cell area where RTx has more probability to save energy compared to DTx. We analyze how to compute such area and propose a simplified model for fast and accurate performance evaluation. Such analysis is described in detail since it will be used as basis for the energy and interference analysis, presented in Sections \ref{sec:EEA} and \ref{sec:IEA}.
%the relaying efficiency with regards to geographical concerns and we propose a model for it. To do so, we define the Relay Efficiency Area (REA) of a relay-aided network in shadowing environment. It aims at characterizing the probability of efficient relaying for each user position, i.e. the probability that relaying allows energy savings compared to direct transmissions.

%uplink transmission and other decode-forward schemes as well.
%Such analysis may first look as a digression to our topic. However, knowing the probability of relaying is necessary to compute the average energy consumption. Moreover, it can be used as a metric for energy efficiency, as we will discuss in paragraph \ref{sec:discussion_prob}.

\subsection{A new definition of the Relay Efficiency Area}

For a given channel realization $(h_d,h_b,h_r)$, the area covered by a relay-aided cell can be divided into two geographical regions, depending on whether DTx or RTx should be performed. The wider is the cell area served by the relay, the more efficient can a relay station be considered.
%Various geometrical patterns have been used in the literature to characterize the area for relaying, e.g. \cite{Chandwani2010, Schober2011}. However, such patterns are neither throughput- nor energy-optimized and do not reflect the performance gains obtained by relaying.
As highlighted in the introduction, various models are considered in the literature to characterize the serving area of a relay, but none are energy-optimized or account for shadowing.
Here, we extend the definition of the pathloss-only  Relay Efficiency Area (REA) in \cite{Journal2} to account also for shadowing environment. We define the REA as the set of all user locations for which relaying is \textit{statistically} more energy-efficient than DTx or is necessary to satisfy a given \textit{outage requirement} $\mathbb{P}_\text{out}$. Both coverage extension and energy gains are included in the REA definition.

\begin{definition} 
The Relay Efficiency Area (REA) of a network in shadowing environment is defined by the pair $(\mathcal{A}_{\mathsf{R}},\mathbb{P}_T)$. Any mobile user $M \left(x,y\right)$ within the geographical area $\mathcal{A}_{\mathsf{R}}$ is served by the relay station with at least the probability $\mathbb{P}_T$, either because RTx is more energy-efficient or because DTx is not feasible, i.e.
% to use the relay station to transmit data (because more energy-efficient or because the direct link is in outage):
%\begin{align*}
$\mathcal{A}_{\mathsf{R}} = \left\lbrace
M \left(x,y\right)
\quad \text{s.t.} \quad \mathbb{P}_T \leq \mathbb{P}_{\text{RTx}}
\right\rbrace$,
%\quad \text{and} \quad 1-\mathbb{P}_\text{out} \leq \mathbb{P}_{\text{Tx}}
%\end{align*}
where $\mathbb{P}_{\text{RTx}} \left(x,y\right)$ is the probability for RTx at user $M \left(x,y\right)$.
\label{def:area_probability}
\end{definition}
%The probability for a user to be served by the relay station is a relevant parameter for network resource management and can be used to support techniques for load-balancing and scheduling. %Moreover, we will see in Section ... that the REA also helps deriving results on power distribution and cost of relaying.
The REA depends on the user rate, the relaying coding scheme and the channel radio propagation. As an example, the REA for two-hop relaying is illustrated in Figure \ref{fig:REA_prob} for various $\mathbb{P}_T$ and $\mathsf{E_{2Hop}^{(dsp)}} $. The dotted lines refer to the model for the REA  proposed in Section \ref{sec:modelREA}. Similarly, we can define $\mathbb{P}_{\text{DTx}} \left(x,y\right)$ as the probability for DTx at user $M \left(x,y\right)$, either because DTx is more energy-efficient or because RTx is not feasible. Then, the outage condition is expressed as $1-\mathbb{P}_\text{out} \leq \mathbb{P}_{\text{RTx}}\left(x,y\right) + \mathbb{P}_{\text{DTx}}\left(x,y\right)$ for user $M \left(x,y\right)$.

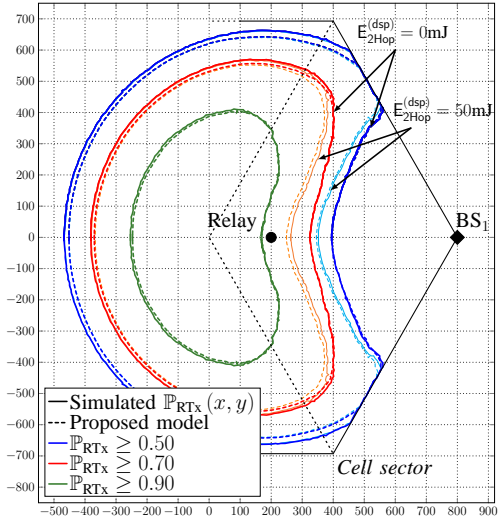
\begin{figure} 
\centering
\centering \resizebox{0.75\columnwidth}{!}{%

\begin{tikzpicture}[scale=4,cap=round,>=latex]

 \pgfplotsset{
    grid style = {
      dash pattern = on 0.05mm off 1mm,
      line cap = round,
      black,
      line width = 0.5pt
    }
  }

  \begin{axis}[%
%    xlabel=Blocks per kernel,%
%    ylabel=Info T/P (Mbps),%
	width=0.92\textwidth,
    height=\textwidth,
    	xmin=-550, xmax=920, ymin=-850, ymax=750,
	%axis equal,
	every axis/.append style={font=\large},  
    grid=major,%
    legend style={at={(axis cs:-530,-830)},anchor=south west, nodes=right, font=\huge},%
    %legend pos={south west},%
    mark size=2.0pt]

\addplot[solid, color=black, line width=1.5]  table[x=xuser ,y=yuser ,col sep=semicolon] {data_REA_figure_000mW_05.txt};
   \addlegendentry{Simulated $\mathbb{P}_{\text{RTx}} \left(x,y\right)$}
     \addplot[dashed, color=black, line width=1.5]  table[x=xuser ,y=yuser ,col sep=semicolon] {data_REA_figure_000mW_05model.txt};
 \addlegendentry{Proposed model}
     
\addplot[solid, color=blue, line width=1.5]  table[x=xuser ,y=yuser ,col sep=semicolon] {data_REA_figure_000mW_05.txt}; 
 \addlegendentry{$\mathbb{P}_{\text{RTx}} \geq 0.50$}

	\addplot[solid, color=red, line width=1.5]  table[x=xuser ,y=yuser ,col sep=semicolon] {data_REA_figure_000mW_07.txt};
	 \addlegendentry{$\mathbb{P}_{\text{RTx}} \geq 0.70$}

	\addplot[solid, color=OliveGreen, line width=1.5]  table[x=xuser ,y=yuser ,col sep=semicolon] {data_REA_figure_000mW_09.txt};   
	 \addlegendentry{$\mathbb{P}_{\text{RTx}} \geq 0.90$}

%%%%%%%%%%%%%%%%%%%%%%%%%%%%%%%

      \addplot[solid, color=ProcessBlue, line width=1]  table[x=xuser ,y=yuser ,col sep=semicolon] {data_REA_figure_050mW_05.txt};
    \addplot[solid, color=ProcessBlue, line width=1]  table[x=xuser ,y=yuser2 ,col sep=semicolon] {data_REA_figure_050mW_05.txt};

     \addplot[dashed, color=ProcessBlue, line width=1]  table[x=xuser ,y=yuser ,col sep=semicolon] {data_REA_figure_050mW_05model.txt};
    \addplot[dashed, color=ProcessBlue, line width=1]  table[x=xuser ,y=yuser2 ,col sep=semicolon] {data_REA_figure_050mW_05model.txt};

\addplot[solid, color=blue, line width=1.5]  table[x=xuser ,y=yuser ,col sep=semicolon] {data_REA_figure_000mW_05.txt}; 
    \addplot[solid, color=blue, line width=1.5]  table[x=xuser ,y=yuser2 ,col sep=semicolon] {data_REA_figure_000mW_05.txt};

     \addplot[dashed, color=blue, line width=1.5]  table[x=xuser ,y=yuser ,col sep=semicolon] {data_REA_figure_000mW_05model.txt};
    \addplot[dashed, color=blue, line width=1.5]  table[x=xuser ,y=yuser2 ,col sep=semicolon] {data_REA_figure_000mW_05model.txt};

	\addplot[solid, color=Orange, line width=1]  table[x=xuser ,y=yuser ,col sep=semicolon] {data_REA_figure_050mW_07.txt};
    \addplot[solid, color=Orange, line width=1]  table[x=xuser ,y=yuser2 ,col sep=semicolon] {data_REA_figure_050mW_07.txt};

	\addplot[dashed, color=orange, line width=1]  table[x=xuser ,y=yuser ,col sep=semicolon] {data_REA_figure_050mW_07model.txt};  
    \addplot[dashed, color=orange, line width=1]  table[x=xuser ,y=yuser2 ,col sep=semicolon] {data_REA_figure_050mW_07model.txt};

	\addplot[solid, color=red, line width=1.5]  table[x=xuser ,y=yuser ,col sep=semicolon] {data_REA_figure_000mW_07.txt};
    \addplot[solid, color=red, line width=1.5]  table[x=xuser ,y=yuser2 ,col sep=semicolon] {data_REA_figure_000mW_07.txt};

	\addplot[dashed, color=red, line width=1.5]  table[x=xuser ,y=yuser ,col sep=semicolon] {data_REA_figure_000mW_07model.txt};  
    \addplot[dashed, color=red, line width=1.5]  table[x=xuser ,y=yuser2 ,col sep=semicolon] {data_REA_figure_000mW_07model.txt};

    \addplot[solid, color=OliveGreen, line width=1.5]  table[x=xuser ,y=yuser2 ,col sep=semicolon] {data_REA_figure_000mW_09.txt};
    
    \addplot[dashed, color=OliveGreen , line width=1.5]  table[x=xuser ,y=yuser ,col sep=semicolon] {data_REA_figure_000mW_09model.txt};   
    \addplot[dashed, color=OliveGreen, line width=1.5]  table[x=xuser ,y=yuser2 ,col sep=semicolon] {data_REA_figure_000mW_09model.txt};

	\draw {(axis cs:400, 692.82) -- (axis cs:0,0) -- (axis cs:400,- 692.82)} [line width=1.25pt, loosely dotted];

	\addplot[line width=1pt, color=black]
    coordinates {
    (100, 692.82)(400, 692.82)(800,0)
	(400,-692.82)(100,- 692.82)
    };
        
     \draw {(axis cs:100, 692.82) -- (axis cs:0, 692.82)} [line width=1pt, loosely dotted];
     \draw {(axis cs:100,- 692.82) -- (axis cs:0,- 692.82)} [line width=1pt, loosely dotted];
    
    \coordinate (center) at (axis cs:0,0);
  \coordinate (bs1) at (axis cs:800,0);
	\coordinate (rs1) at (axis cs:200,0);
%  \coordinate (rs11) at (axis cs:433.0127,-250);
%  
	\draw (bs1) node[diamond, draw, fill = black,scale=1] {};
%	\draw (bs2) node[diamond, draw, fill = black,scale=0.5] {};
%	\draw (bs3) node[diamond, draw, fill = black,scale=0.5] {};
%	
	\draw (rs1) node[circle, draw, fill = black,scale=1] {};
%	\draw (rs11) node[circle, draw, fill = black,scale=0.5] {};
	
  \draw (axis cs:850,50) node {\huge BS$_1$};
  \draw (axis cs:75,50) node {\huge Relay};	
  \draw (axis cs:400,-700) node [anchor=north west] { \huge \textit{Cell sector}};	
  
  \draw[<-,line width=1.5pt] (axis cs:520,350) -- (axis cs:600,600) node [anchor=south]{ \hspace*{10pt} \LARGE $\mathsf{E_{2Hop}^{(dsp)}} = 0$mJ};
  \draw[<-,line width=1.5pt] (axis cs:400,400) -- (axis cs:600,600);
  
  \draw[<-,line width=1.5pt] (axis cs:390,150) -- (axis cs:650,350) node [anchor=south west]{   \hspace*{-25pt} \LARGE $\mathsf{E_{2Hop}^{(dsp)}} = 50$mJ};
 \draw[<-,line width=1.5pt] (axis cs:350,250) -- (axis cs:650,350);  
  
%  
%  \draw (axis cs:-550,750) node [anchor=west] {\huge $\mathbb{P}_{\text{RTx}} \geq 0.50$};
%  \draw (axis cs:-550,675) node [anchor=west] {\huge $\mathbb{P}_{\text{RTx}} \geq 0.70$};
%  \draw (axis cs:-550,600) node [anchor=west] {\huge $\mathbb{P}_{\text{RTx}} \geq 0.90$};

    \end{axis}

\end{tikzpicture}
}
\caption{Relay Efficiency Area: simulation and model (Two-hop relaying, $N_r=1$, $\mathsf{D}_r =700$m)}
\label{fig:REA_prob}
\end{figure}

%%%%%%%%%%%%%%%%%%%%%%%%%
\newcounter{MYtempeqncnt}
\begin{figure*}[!t]
\setcounter{MYtempeqncnt}{\value{equation}}
\setcounter{equation}{8}
{\small
\begin{align}
\mathbb{P}_{\text{low}}^{(1)} = & \max \left( 0,
\mathbb{P} \left( E_{b+r} \leq E_d \right) \; - \left[ \;
\mathbb{P} \left( \mathsf{E_R^{(m)}} \leq E_d \leq \mathsf{E_B^{(m)}} +\mathsf{E_R^{(m)}} \right)
\mathbb{P} \left( E_{b+r} \leq \mathsf{E_B^{(m)}} +  \mathsf{E_R^{(m)}} \right)
%\nonumber \right. \right.\\ & \left. \left.
+ \mathbb{P} \left( \mathsf{E_B^{(m)}} +\mathsf{E_R^{(m)}} \leq E_d \right)
\right]
\right) \label{eq:Prob_low1}
\\
\mathbb{P}_{\text{low}}^{(2)} = &  \mathbb{P} \left( \mathsf{E_R^{(m)}} \leq E_d \leq \mathsf{E_B^{(m)}} \right) \mathbb{P} \left( E_{b+r} \leq \mathsf{E_R^{(m)}} \right)
\label{eq:Prob_low2}
\end{align}
}
\hrulefill
\setcounter{equation}{\value{MYtempeqncnt}}
\end{figure*}
%%%%%%%%%%%%%%%%%%%%%%%%%%%%%%%%%%%%%%%%%%%%%%%%%%%%%

\subsection{Characterization of the relaying probability $\mathbb{P}_{\text{RTx}} \left(x,y\right)$}
\label{sec:modelREA}

For readability, the relaying probability is analyzed for downlink two-hop relaying only, but 
%we will show in Section \ref{sec:validation} that
is valid for other scenarios, as discussed in Section \ref{sec:discussion_prob}.
We first focus on the case $\mathsf{E^{(dsp)}} = 0$.  

\subsubsection{Notation for the energy consumption}

We define $E_{d}$ as the RF transmit energy consumed by the BS when DTx is used. Similarly, $E_{b}$ and $E_{r}$ stand for the energy consumed by the BS and RS respectively when RTx is used.
Considering Gaussian signaling, we deduce the energy consumption based on capacity formulas. To send data to user $M \left(x,y\right)$, we have:
\begin{align}
\left \lbrace
\begin{array}{l}
E_{d} = \left(2^{\mathcal{R}}-1\right) N \frac{\gamma_{d}\left[ x,y\right]}{s_d} \\
E_{b} = \left(2^{2\mathcal{R}}-1\right) N \frac{\gamma_{b}\left[ \mathsf{D}_r \right]}{2 s_b} \\
E_{r} = \left(2^{2\mathcal{R}}-1\right)  N \frac{\gamma_{r}\left[ x,y \right]}{2 s_r} ,
\end{array} \right. 
\label{eq:E_i}
\end{align}
where $N$ is the variance of the AWGN, $s_k$ is the shadowing coefficient and $\gamma_k$ the path-loss, as given in Eq. \eqref{eq:pathloss}.
Then, we use superscript $^{(0)}$ to refer to the non-shadowing case, when only path-loss is considered (i.e. the variance of the shadowing coefficient $s_k$ is $\sigma_k^2=0$). Assuming log-normal shadowing environment, we have
\begin{align}
E_{k} = \frac{E_{k}^{(0)}}{s_k} \sim \log\mathcal{N} \left(\mu_k, \sigma_k^2\right)
\quad \text{with} \; \; \mu_k = \ln \left(E_{k}^{(0)} \right).
\label{eq:E_i_0}
\end{align}

\subsubsection{Probability for energy-efficient relaying}

For a given channel realization, a transmission is relayed if DTx is not feasible or if RTx is more energy-efficient. This implies:
\begin{align}
\mathbb{P}_{\text{RTx}} \left(x,y\right) = \mathbb{P}_{\text{CR}}\left(x,y\right) +\mathbb{P}_{\text{ER}}\left(x,y\right)
\end{align}
where $\mathbb{P}_{\text{CR}}$ and $\mathbb{P}_{\text{ER}}$ are as defined below. As we will see in next section, these probabilities are necessary to define a model for the energy consumption and interference.

\begin{definition}
When the user-BS link is weak and data cannot be sent using DTx given the channel realization, relaying is performed to extend the cell coverage and maintain the outage requirement. This is referred to as the Coverage Condition for Relaying (CR), which occurs with probability $\mathbb{P}_{\text{CR}}$. Given the energy constraints $\mathsf{E_B^{(m)}}$ at BS and $\mathsf{E_R^{(m)}}$ at RS, we have
\begin{align*}
\mathcal{C}_{\text{CR}} = & \; \left \lbrace 
  \;E_d > \mathsf{E_B^{(m)}}  \;\; \cap \;\;
E_b \leq \mathsf{E_B^{(m)}}   \;\; \cap \;\;
E_r \leq \mathsf{E_R^{(m)}}
\right\rbrace.
\end{align*}
 %The CR-condition .
\label{def:CR_condition} 
\end{definition}
\vspace*{-10pt}
Since energy consumption in each scenario is often independent, 
 $\mathbb{P}_{\text{CR}}$ can generally be computed as
 \begin{align*}
 \mathbb{P}_{\text{CR}} = \mathbb{P} \left( E_d > \mathsf{E_B^{(m)}} \right) \mathbb{P} \left( E_b \leq \mathsf{E_B^{(m)}} \right) \mathbb{P} \left( E_r \leq \mathsf{E_R^{(m)}} \right).
 \end{align*}
Similarly, we define the Coverage Condition for Direct transmission (CD) for which data can be sent only via the direct link. It occurs with probability $\mathbb{P}_{\text{CD}}$ and is generally computed in closed-form as for  $\mathbb{P}_{\text{CR}}$.

% as expressed in Section \ref{sec:validation} for both two-hop transmission and repetition-coded full decode-forward. %Moreover, the link between the relay and the base station is usually very strong, such that $\mathbb{P} \left( E_r \leq \mathsf{E_R^{(m)}} \right) \simeq 1 $.

\begin{definition}
When both DTx and RTx are feasible, relaying is performed if it is more energy-efficient. This defines the Energy-Efficient Condition for Relaying (ER), expressed as: 
\begin{align*}
\mathcal{C}_{\text{ER}} = & \; \left \lbrace 
E_{k} \leq \mathsf{E_B^{(m)}}   \;\; \cap \;\;
E_{r} \leq \mathsf{E_R^{(m)}}  \;\; \cap \;\;
E_{b} + E_{r} \leq E_{d} \;
\right\rbrace.
\end{align*}
with $k \in\left\lbrace d,b\right\rbrace$. The ER-condition has probability $\mathbb{P}_{\text{ER}}$. 
\label{def:ER_condition}
\end{definition}

Similarly, we define the Energy-Efficient Condition for Direct transmission (ED) for which DTx is more energy-efficient. It occurs with probability $\mathbb{P}_{\text{ED}}$:
\begin{align}
\mathbb{P}_{\text{ED}} = \mathbb{P} & \left( 
E_{\lbrace d,b \rbrace} \leq \mathsf{E_B^{(m)}}
\; \cap \; E_r \leq \mathsf{E_R^{(m)}}
 \right)
 - \mathbb{P}_{\text{ER}}.
 \label{eq:Prob_ED}
\end{align}
Thus, the outage condition at user $M \left(x,y\right)$ is given by
\begin{align*}
1-\mathbb{P}_\text{out} \leq & \; \mathbb{P}_{\text{ER}} + \mathbb{P}_{\text{CR}} + \mathbb{P}_{\text{ED}}+ \mathbb{P}_{\text{CD}}
\\
\leq & \; \mathbb{P} \left( 
E_{\lbrace d,b \rbrace} \leq \mathsf{E_B^{(m)}}
\; \cap \; E_r \leq \mathsf{E_R^{(m)}}
 \right) + \mathbb{P}_{\text{CR}} + \mathbb{P}_{\text{CD}}
\end{align*}

Due to power constraints, $\mathbb{P}_{\text{ER}}$ is obtainable by a triple integral over $s_d$, $s_b$ and $s_r$, as in Eq. \eqref{eq:triple_integral} of Appendix A. Thus, contrary to $\mathbb{P}_{\text{CR}}$, it is not separable and may not exist in closed form. Even tough numerical computation  can be envisaged using mathematical software, 
this approach rapidly becomes unsuitable for large networks, even considering the simple two-hop scheme.
%+ if results needed as an average over user location or relay (as we will do in this work), it means computing a succession of 5 or 6 integrals

%In the following paragraphs, we aim at computing $\mathbb{P}_{\text{CR}}\left(x,y\right)$ and $\mathbb{P}_{\text{ER}}\left(x,y\right)$ to characterize the Relay Efficiency Area.

\subsection{Proposed model for the REA}

To address the key issue of computation, we propose a closed-form lower bound for $\mathbb{P}_{\text{ER}}$. %We then deduce an upper-bound for the energy consumption.

%\textit{VERY IMPORTANT PARAGRAPH:} in most of cases, the ER-probability $\mathbb{P}_{\text{ER}} $ cannot be computed in closed-form, we thus propose to bound it. Moreover, regarding log-normal shadowing, we need to "isolate" the sum $P_{s} + P_{r}$ to allow closed-form computation of upper and lower bounds.

\begin{lemma}
The probability $\mathbb{P}_{\text{ER}}$ for energy-efficient relaying is lower-bounded by the sum  $\mathbb{P}_{\text{low}} = \mathbb{P}_{\text{low}}^{(1)} + \mathbb{P}_{\text{low}}^{(2)}$ where $\mathbb{P}_{\text{low}}^{(1)}$ and $\mathbb{P}_{\text{low}}^{(2)}$ are as given in Eq. \eqref{eq:Prob_low1} and \eqref{eq:Prob_low2} at the top of page. \addtocounter{equation}{2}
Contrary to $\mathbb{P}_{\text{ER}}$, the lower-bound $\mathbb{P}_{\text{low}}$  consisting of elementary probabilities that can be individually computed or have closed-form approximations.
% Eq. \eqref{eq:Prob_low1} and \eqref{eq:Prob_low2} at the top of the page.
\label{lemma:P_low}
\end{lemma}
\begin{proof}
See Appendix \ref{app:lower_bound}.
\end{proof}
Also note that, $\mathbb{P}_{\text{low}}$ gives an upper-bound for $\mathbb{P}_{\text{ED}}$:
\begin{align}
\mathbb{P}_{\text{ED}} \leq \mathbb{P} & \left( 
E_{\lbrace d,b \rbrace} \leq \mathsf{E_B^{(m)}}
\; \cap \; E_r \leq \mathsf{E_R^{(m)}}
 \right)
 - \mathbb{P}_{\text{low}}.
 \label{eq:upper_bound_P_ED}
\end{align}

\begin{model}[\textbf{REA}]
The Relay Efficiency Area, as characterized in Definition \ref{def:area_probability}, is modeled in log-normal outdoors propagation environments by the pair $( \widehat{\mathcal{A}_{\mathsf{R}}}, \mathbb{P}_T)$, where
\begin{align*}
\widehat{\mathcal{A}_{\mathsf{R}}} = \left\lbrace
M \left(x,y\right) \quad \text{s.t.} \right.
 \left.\mathbb{P}_T \leq \mathbb{P}_{\text{low}} \left(x,y\right) + \mathbb{P}_{\text{CR}}\left(x,y\right)
\right\rbrace .
\end{align*}
$\mathbb{P}_{\text{low}}$ and $\mathbb{P}_{\text{CR}}$ are given by Lemma \ref{lemma:P_low} and Definition \ref{def:CR_condition} respectively.
\label{model:area_prob_model}
\end{model}
Any mobile user $M \left(x,y\right)$ circulating within $(\widehat{\mathcal{A}_{\mathsf{R}}},\mathbb{P}_T)$ is also within $({\mathcal{A}_{\mathsf{R}}},\mathbb{P}_T)$, such that it has at least the probability $\mathbb{P}_T$ to save energy via relaying.% The model for the REA is depicted in Figure \ref{fig:REA_prob}.
%\begin{proof}
%This inclusion directly follows from Lemma \ref{lemma:P_low}.
%\end{proof}

\subsection{Accounting for the circuitry consumption}

Up to now, we have focused on the case $\mathsf{E^{(dsp)}} = 0$, i.e. accounting only for the transmit energy. These results, however, can be generalized to show  the overall energy consumption. In this case, data is relayed if the overall energy consumed using RTx is less than the consumption using DTx. As defined in Section \ref{sec:coding_scheme}, the overall consumption includes the RF amplifier efficiency (given by $\eta_R$ and $\eta_B$) and the additional energy dissipated for decoding and re-encoding at the relay station (given by $\mathsf{E^{(dsp)}}$). The energy-efficient condition for relaying becomes:
\begin{align*}
\mathcal{C}_{\text{ER}}^\circ = 
& \; \left \lbrace 
E_{b}^\circ + E_{r}^\circ + \mathsf{E^{(dsp)}} \leq E_{d}^\circ \;
\;\; \cap \;\; 
E^\circ_{\left\lbrace d,b\right\rbrace} \leq \eta_B \mathsf{E_B^{(m)}}  
\;\; \cap \;\;
\right. \\ & \; \left.
 E^\circ_{r} \leq \eta_R \mathsf{E_R^{(m)}}
\right\rbrace 
\; \text{with }\; \; 
E_{k}^\circ \sim \log\mathcal{N} \left(\mu_k + \ln \left(\eta_k \right), \sigma_k^2\right)
\end{align*} 
%with $E_{k}^\circ \sim \log - \mathcal{N} \left(\mu_k + \ln \left(\eta_k \right), \sigma_k^2\right)$.
% and $\mathsf{E^{(dsp)}}$ equal to $\mathsf{E_{2Hop}^{(dsp)}}$ or $\mathsf{E}_\mathsf{{PDF}}^{(dsp)}$ depending on the considered coding scheme.
\begin{model}[\textbf{REA, extended}]
The REA accounting for the overall energy consumption is modelled by $(\widehat{\mathcal{A}^{\circ}_{\mathsf{R}}}, \mathbb{P}_T)$, with
\begin{align*}
\widehat{\mathcal{A}^{\circ}_{\mathsf{R}}} = \left\lbrace
M \left(x,y\right) \quad \text{s.t.} \quad \mathbb{P}_T \leq \mathbb{P}_{\text{low}}^{\circ} \left(x,y\right) + \mathbb{P}_{\text{CR}}^{\circ}\left(x,y\right)
\right\rbrace .
\end{align*}
$\mathbb{P}_{\text{low}}^{\circ}$ and $\mathbb{P}_{\text{CR}}^{\circ}$ are given by Lemma \ref{lemma:P_low} and Definition \ref{def:CR_condition} respectively, but with the following replacement:
\begin{align}
\begin{array}{r l}
E_k \longrightarrow & E_{k}^\circ, \quad k \in \left\lbrace d,b,r \right\rbrace 
\\
E_{b+r} \longrightarrow & E_{b}^\circ + E_{r}^\circ + \mathsf{E^{(dsp)}}
\\
\mathsf{E_B^{(m)}} \longrightarrow & \eta_B \mathsf{E_B^{(m)}} 
\\
\mathsf{E_R^{(m)}} \longrightarrow & \eta_R \mathsf{E_R^{(m)}} + \mathsf{E^{(dsp)}}
\end{array}
\label{eq:model_extension_subs}
\end{align}
\end{model}
This model is illustrated in Figure \ref{fig:REA_prob} for $\mathsf{E_{2Hop}^{(dsp)}} =$ 50mJ.% The model is less tight when the dissipated energy at the relay station is high. However, we point out that the dissipated energy largely dominates the overall consumption when $\mathsf{E^{(dsp)}} =$ 150mW.

\subsection{Discussion}
\label{sec:discussion_prob}

%\subsubsection*{Computational efficiency}
%To address the key issue of computation, we propose a lower bound for $\mathbb{P}_{\text{ER}}$ which consists of elementary probabilities that can be individually computed or have closed-form approximations.
%
%
%Due to power constraints, the ER-probability can only be written as a triple integral over $s_d$, $s_b$ and $s_r$, as expressed in Eq. \eqref{eq:triple_integral} of Appendix A. So, contrary to $\mathbb{P}_{\text{CR}}$, it is not closed-form. Though numerical computation  can be envisaged by using mathematical software, 
%such approach can rapidly become inappropriate to estimate the relaying probability in a large network, even considering the simple two-hop scheme.

\subsubsection{Model validity}

The proposed model has been validated under several outdoors environment settings and for both uplink and downlink transmissions. For the uplink, we generally have $\mathsf{E_U^{(m)}} \leq \mathsf{E_R^{(m)}}$ for the user and relay energy constraints, implying that $\mathbb{P}_{\text{low}}^{(2)} =0$.
This model can be also used for any DF schemes for which the consumed energies $E_d$, $E_r$ and $E_b$ are log-normally distributed, such as the repetition-coded full DF scheme in \cite{laneman2004} where the user decodes data using maximum ratio combining on the signal received from the BS and RS during both phases.

%
%Assuming log-normal shadowing, we have
%\begin{align}
%E_{k} = \frac{E_{k}^{(0)}}{s_k} \sim \log - \mathcal{N} \left(\mu_k, \sigma_k^2\right),
%\label{eq:E_i_0}
%\end{align}
%where the mean $\mu_k = \ln \left(E_{k}^{(0)} \right)$.

%Let's consider as another example repetition-coded full decode-forward. 
%This scheme is similar to Two-hop relaying, but here, the user decodes data using maximum ratio combining on the signal received from the BS and the relay during both phases. The power allocation minimizes the total energy consumption and, assuming CSIT on both $h_d$ and $h_b$ at the relay, we have:
%\begin{align}
%\left \lbrace
%\begin{array}{rl}
%E_{b} & = \left(2^{2\mathcal{R}}-1\right) N \frac{\gamma_{b}\left[ x,y\right]}{2 s_b}  \\
%E_{r} & = \left(2^{2\mathcal{R}}-1\right) N \frac{\gamma_{r}\left[ x,y \right]}{2 s_r}  \left( 1- \frac{\gamma_{b}\left[ \mathsf{D}_r\right]}{\gamma_{d}\left[ x,y\right]} \frac{s_d}{s_b}\right) .
%\end{array} \right.
%\end{align}
%The analysis proposed in this work is still valid for this coding scheme since $E_b$ is the product of two log-normal random variables and is thus as well log-normally distributed.

\subsubsection{Model utilization}
The primary use of the REA is to compute the spatial distribution of transmit energy consumption, as described in the next section. Thanks to the proposed model, many relay configurations and propagation environments can be analyzed in a reasonable time, which even allows finding the optimal relay location by exhaustive search.
Moreover, this model can help decide between DTx and RTx when only statistics of the channel realizations are known at the transmitters (partial CSIT). In this case, the proposed model specifies for each user location the path which has the highest probability to save energy.
% For example, if at $M \left(x,y\right)$, $\mathbb{P}_{\text{ER}}=$ 70\%, relaying will be more energy-efficient than DTx for 70\% of channel realizations.

Next, the concept of Efficiency Area has wide application since it is based on the network geometry and ensures a minimum performance. Such a framework allows the deployment of relay stations in an efficient manner, simply by locating relays such that hotspots or cell regions with poor performance are included within the corresponding Efficiency Area. This concept can be particularly applicable to non-uniform random user locations, e.g. a hotspot-type distribution. More specifically, the probability for energy-efficient relaying offers valuable support for network resource management, relay selection, load-balancing, scheduling or BS switch off. For example, switch-off can be decided if the probability to reduce energy consumption via relaying is above a given threshold $\mathbb{P}_T$, i.e. if the considered cell area is included in $(\mathcal{A}_{\mathsf{R}},\mathbb{P}_T)$.

%\begin{figure} 
%\centering \input{figure_2.tex}
%\caption{System model for a hexagonal cell aided by 6 relays}
%\label{fig:cell_system}
%\end{figure}

%%%%%%%%%%%%%%%%%%%%%%%%%%%%%%%%%%%%%%%%%%%%%%%%%%%%%%%%%%%%%%%%%%%%%%%%%%%%%%%

\section{Energy Consumption and Energy Efficiency Area with shadowing}
\label{sec:EEA}

%%%%%%%%%%%%%%%%%%%%%%%%%
%\newcounter{MYtempeqncnt}
\begin{figure*}[!t]
\setcounter{MYtempeqncnt}{\value{equation}}
\setcounter{equation}{15}
{\small
\begin{align}
\mathsf{E}_{\text{up}}^{\text{(ED)}} = & \;
g(d,\mathsf{E_B^{(m)}}) \mathbb{P} \left( E_d \leq \mathsf{E_B^{(m)}} \right)
\mathbb{P} \left( E_b \leq \mathsf{E_B^{(m)}} \; \cap \; E_r \leq \mathsf{E_R^{(m)}} \right)
 - \exp \left( \mu_{d} + \frac{\sigma_{d}^2}{2} \right) \mathbb{P}_{\text{low}}^{(1,\text{ED})}
 \nonumber \\ &
 - \exp \left( \mu_d + \frac{\sigma_d^2}{2} \right)
\frac{\Phi \left(- \sigma_d + \frac{\ln \left(\mathsf{E_R^{(m)}} \right) - \mu_d}{\sigma_d}  \right)-\Phi \left(- \sigma_d + \frac{\ln \left(\mathsf{E_B^{(m)}} \right) - \mu_d}{\sigma_d}  \right)}
{\Phi \left(\frac{\ln \left(\mathsf{E_B^{(m)}} \right) - \mu_d}{\sigma_d} \right)-\Phi \left(\frac{\ln \left(\mathsf{E_R^{(m)}} \right) - \mu_d}{\sigma_d} \right)}
\mathbb{P}_{\text{low}}^{(2)}
\label{eq: E_ED}
\end{align}
}
\hrulefill
\setcounter{equation}{\value{MYtempeqncnt}}
\end{figure*}
%%%%%%%%%%%%%%%%%%%%%%%%%%%%%%%%%%%%%%%%%%%%%%%%%%%%%

We analyze the spatial distribution of the energy consumption within the cell and define the Energy Efficiency Area (EEA) as the cell area for which the average transmit energy consumption does not exceed $\mathsf{E}_T$. 
The EEA is related to the maximum energy necessary to transmit data at a given rate. Such a metric is relevant for performance analysis and encompasses fairness between the served users. Indeed, a relay station cannot be energy-efficient for users located in a certain area of the cell only, while showing unacceptably high energy consumption elsewhere.

%\subsection{Analysis of the transmit energy}

%As for the REA, we consider at first only the transmit energy consumption.

\subsection{Definition of the EEA}

The average transmit energy, denoted as $\mathbb{E} \left[E^{\mathsf{(RF)}}\right]$, consumed to send data to a given user $M \left(x,y\right)$, is equal to the sum of the energy consumed by BS only when DTx is used, and by both BS and RS when RTx is used. It is expressed as follows, with $E_{b+r} = E_b + E_r$:
\begin{align}
\mathbb{E} \left[E^{\mathsf{(RF)}}\right] &=
	\mathbb{P}_{\text{CR}} \, \mathbb{E}\left[E_{b+r} \, \vert \,  \mathcal{C}_{\text{CR}}\right]
	+ \mathbb{P}_{\text{ER}} \, \mathbb{E}\left[E_{b+r} \, \vert \, \mathcal{C}_{\text{ER}}\right]
	\nonumber  \\ 	&
	+ \mathbb{P}_{\text{CD}} \, \mathbb{E}\left[E_d \, \vert \, \mathcal{C}_{\text{CD}}\right]
	+ \mathbb{P}_{\text{ED}} \, \mathbb{E}\left[E_d \, \vert \, \mathcal{C}_{\text{ED}}\right].
	\label{eq:average_energy}
\end{align}

\begin{definition} 
The Energy Efficiency Area (EEA) of a network in shadowing environment defines the energy range across the cell. It is characterized by the pair $(\mathcal{A}_{\mathsf{E}},\mathsf{E}_T)$: the average transmit energy consumed by any mobile user $M \left(x,y\right)$ within $\mathcal{A}_{\mathsf{E}}$ does not exceed $\mathsf{E}_T$, i.e.
\begin{align*}
\mathcal{A}_{\mathsf{E}} = \left\lbrace
M \left(x,y\right) \quad \text{s.t.} \quad \mathbb{E} \left[E^{\mathsf{(RF)}}\right] \leq \mathsf{E}_T
\right\rbrace.
\end{align*} 
\label{def:area_energy}
\end{definition}
\vspace*{-15pt}
Figure \ref{fig:energy_consumption} illustrates $\mathcal{A}_{\mathsf{E}}$ and its corresponding model for various values of $\mathsf{E}_T$. Also note that the areas for relaying probability $\mathcal{A}_{\mathsf{R}}$ and for energy $\mathcal{A}_{\mathsf{E}}$ are distinct. Next, we describe the proposed model for the EEA.

\begin{figure}
\centering
\centering \resizebox{0.75\columnwidth}{!}{%

\begin{tikzpicture}[scale=4,cap=round,>=latex]

 \pgfplotsset{
    grid style = {
      dash pattern = on 0.05mm off 1mm,
      line cap = round,
      black,
      line width = 0.5pt
    }
  }

  \begin{axis}[%
%    xlabel=Blocks per kernel,%
%    ylabel=Info T/P (Mbps),%
	width=0.92\textwidth,
    height=\textwidth,
    	xmin=-550, xmax=920, ymin=-850, ymax=750,
	%axis equal,
	every axis/.append style={font=\large},  
    grid=major,%
     legend style={at={(axis cs:-530,-830)},anchor=south west, nodes=right, font=\huge},%
    %legend pos={south west},%
    mark size=2.0pt]

\addplot[solid, color=black, line width=1.5]  table[x=xuser ,y=yuser ,col sep=semicolon] {data_EEA_figure_000mW_01.txt};
   \addlegendentry{Simulated $\mathbb{E} \left[E^{\mathsf{(RF)}}\right]$}
     \addplot[dashed, color=black, line width=1.5]  table[x=xuser ,y=yuser ,col sep=semicolon] {data_EEA_figure_000mW_01model.txt};
 \addlegendentry{Proposed model}
     
%     
%     
%\addplot[solid, color=blue, line width=1.5]  table[x=xuser ,y=yuser ,col sep=semicolon] {data_EEA_figure_000mW_01.txt}; 
% \addlegendentry{$\mathbb{E} \left[E^{\mathsf{(RF)}}\right] \leq 0.1$}
%
%	\addplot[solid, color=ProcessBlue, line width=1.5]  table[x=xuser ,y=yuser ,col sep=semicolon] {data_EEA_figure_000mW_02.txt};
%	 \addlegendentry{$\mathbb{E} \left[E^{\mathsf{(RF)}}\right] \leq 0.2$}
%
%	\addplot[solid, color=OliveGreen, line width=1.5]  table[x=xuser ,y=yuser ,col sep=semicolon] {data_EEA_figure_000mW_05.txt};   
%	 \addlegendentry{$\mathbb{E} \left[E^{\mathsf{(RF)}}\right] \leq 0.5$}
%
%\addplot[solid, color=Orange, line width=1.5]  table[x=xuser ,y=yuser ,col sep=semicolon] {data_EEA_figure_000mW_08.txt};   
%	 \addlegendentry{$\mathbb{E} \left[E^{\mathsf{(RF)}}\right] \leq 0.8$}
%
%
%\addplot[solid, color=red, line width=1.5]  table[x=xuser ,y=yuser ,col sep=semicolon] {data_EEA_figure_000mW_1.txt};   
%	 \addlegendentry{$\mathbb{E} \left[E^{\mathsf{(RF)}}\right] \leq 1$}
%

%%%%%%%%%%%%%%%%%%%%%%%%%%%%%%%
    \addplot[solid, color=Plum, line width=1.5]  table[x=xuser ,y=yuser ,col sep=semicolon] {data_EEA_figure_000mW_0075.txt};
    \addplot[solid, color=Plum, line width=1.5]  table[x=xuser ,y=yuser2 ,col sep=semicolon] {data_EEA_figure_000mW_0075.txt};
        
     \addplot[dashed, color=Plum, line width=1.5]  table[x=xuser ,y=yuser ,col sep=semicolon] {data_EEA_figure_000mW_0075model.txt};
    \addplot[dashed, color=Plum, line width=1.5]  table[x=xuser ,y=yuser2 ,col sep=semicolon] {data_EEA_figure_000mW_0075model.txt};

%%%%%%%%%%%%%%%%%%%%%%%%%%%%%%%
    \addplot[solid, color=blue, line width=1.5]  table[x=xuser ,y=yuser ,col sep=semicolon] {data_EEA_figure_000mW_01.txt};
    \addplot[solid, color=blue, line width=1.5]  table[x=xuser ,y=yuser2 ,col sep=semicolon] {data_EEA_figure_000mW_01.txt};
        
     \addplot[dashed, color=blue, line width=1.5]  table[x=xuser ,y=yuser ,col sep=semicolon] {data_EEA_figure_000mW_01model.txt};
    \addplot[dashed, color=blue, line width=1.5]  table[x=xuser ,y=yuser2 ,col sep=semicolon] {data_EEA_figure_000mW_01model.txt};

%%%%%%%%%%%%%%%%%%%%%%%%%%%%%%%    
      \addplot[solid, color=ProcessBlue, line width=1.5]  table[x=xuser ,y=yuser ,col sep=semicolon] {data_EEA_figure_000mW_02.txt};
    \addplot[solid, color=ProcessBlue, line width=1.5]  table[x=xuser ,y=yuser2 ,col sep=semicolon] {data_EEA_figure_000mW_02.txt};
        
     \addplot[dashed, color=ProcessBlue, line width=1.5]  table[x=xuser ,y=yuser ,col sep=semicolon] {data_EEA_figure_000mW_02model.txt};
    \addplot[dashed, color=ProcessBlue, line width=1.5]  table[x=xuser ,y=yuser2 ,col sep=semicolon] {data_EEA_figure_000mW_02model.txt};

   %%%%%%%%%%%%%%%%%%%%%%%%%%%%%%%          
    \addplot[solid, color=OliveGreen, line width=1.5]  table[x=xuser ,y=yuser ,col sep=semicolon] {data_EEA_figure_000mW_05.txt};
    \addplot[solid, color=OliveGreen, line width=1.5]  table[x=xuser ,y=yuser2 ,col sep=semicolon] {data_EEA_figure_000mW_05.txt};

	\addplot[dashed, color=OliveGreen, line width=1.5]  table[x=xuser ,y=yuser ,col sep=semicolon] {data_EEA_figure_000mW_05model.txt};  
    \addplot[dashed, color=OliveGreen, line width=1.5]  table[x=xuser ,y=yuser2 ,col sep=semicolon] {data_EEA_figure_000mW_05model.txt};

   %%%%%%%%%%%%%%%%%%%%%%%%%%%%%%%          
	\addplot[solid, color=Orange, line width=1.5]  table[x=xuser ,y=yuser ,col sep=semicolon] {data_EEA_figure_000mW_08.txt};
    \addplot[solid, color=Orange, line width=1.5]  table[x=xuser ,y=yuser2 ,col sep=semicolon] {data_EEA_figure_000mW_08.txt};

	\addplot[dashed, color=orange, line width=1.5]  table[x=xuser ,y=yuser ,col sep=semicolon] {data_EEA_figure_000mW_08model.txt};  
    \addplot[dashed, color=orange, line width=1.5]  table[x=xuser ,y=yuser2 ,col sep=semicolon] {data_EEA_figure_000mW_08model.txt};

    \addplot[solid, color=red, line width=1.5]  table[x=xuser ,y=yuser ,col sep=semicolon] {data_EEA_figure_000mW_1.txt};
    \addplot[solid, color=red, line width=1.5]  table[x=xuser ,y=yuser2 ,col sep=semicolon] {data_EEA_figure_000mW_1.txt};
    
    \addplot[dashed, color=red , line width=1.5]  table[x=xuser ,y=yuser ,col sep=semicolon] {data_EEA_figure_000mW_1model.txt};   
    \addplot[dashed, color=red, line width=1.5]  table[x=xuser ,y=yuser2 ,col sep=semicolon] {data_EEA_figure_000mW_1model.txt};

	\draw {(axis cs:400, 692.82) -- (axis cs:0,0) -- (axis cs:400,- 692.82)} [line width=1.25pt, loosely dotted];

	\addplot[line width=1pt, color=black]
    coordinates {
    (100, 692.82)(400, 692.82)(800,0)
	(400,-692.82)(100,- 692.82)
    };
        
     \draw {(axis cs:100, 692.82) -- (axis cs:0, 692.82)} [line width=1pt, loosely dotted];
     \draw {(axis cs:100,- 692.82) -- (axis cs:0,- 692.82)} [line width=1pt, loosely dotted];
    
    \coordinate (center) at (axis cs:0,0);
  \coordinate (bs1) at (axis cs:800,0);
	\coordinate (rs1) at (axis cs:200,0);
%  \coordinate (rs11) at (axis cs:433.0127,-250);
%  
	\draw (bs1) node[diamond, draw, fill = black,scale=1] {};
%	\draw (bs2) node[diamond, draw, fill = black,scale=0.5] {};
%	\draw (bs3) node[diamond, draw, fill = black,scale=0.5] {};
%	
	\draw (rs1) node[circle, draw, fill = black,scale=1] {};
%	\draw (rs11) node[circle, draw, fill = black,scale=0.5] {};
	
  \draw (axis cs:850,50) node {\huge BS$_1$};
  \draw (axis cs:220,50) node {\huge Relay};	
  \draw (axis cs:400,-700) node [anchor=north west] {\huge \textit{Cell sector}};

  \draw (axis cs:450,650) node [anchor=west] {\LARGE $\mathbb{E} \left[E^{\mathsf{(RF)}}\right]  \leq 1$J};
  \draw (axis cs:510,530) node [anchor=west] {\LARGE $\leq 0.8$};
  \draw (axis cs:580,400) node [anchor=west] {\LARGE $\leq 0.5$};
  \draw (axis cs:650,270) node [anchor=west] {\LARGE $\leq 0.2$};
  \draw (axis cs:400,200) node [anchor=west] {\LARGE $\leq 0.1$};
    \draw (axis cs:530,100) node [anchor=west] {\LARGE $\leq 0.075$};

    \end{axis}

\end{tikzpicture}
}
\caption{Energy Efficiency Area: simulation and model (Two-hop relaying, $N_r=1$, $\mathsf{D}_r =700$m)}
\label{fig:energy_consumption}
\end{figure}
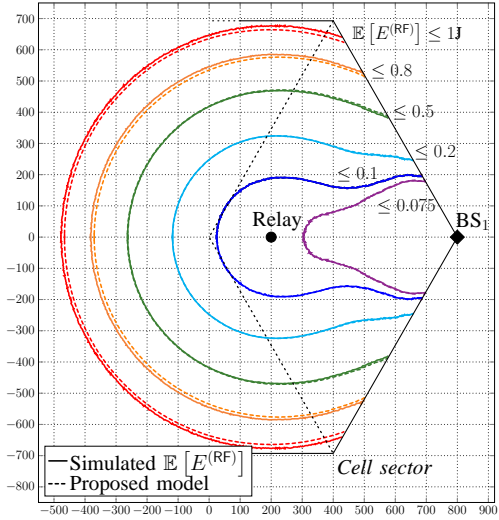

\vspace*{-10pt}
\subsection{Analysis of the EEA with shadowing}

We focus on the computation of $\mathbb{E} \left[E^{\mathsf{(RF)}}\right]$ in Eq. \eqref{eq:average_energy} for log-normal shadowing environments. For the coverage conditions CD and CR, both expectations are computed in closed-form:
\begin{align}
& \left\lbrace
\begin{array}{l}
\mathbb{E}\left[E_{s+r} \, \vert \, \mathcal{C}_{\text{CR}}\right] = 
\underset{j \in \left\lbrace r,b \right\rbrace}{\sum} g(j,\mathsf{E_j^{(m)}}) 
\\
\mathbb{E}\left[E_{d} \, \vert \, \mathcal{C}_{\text{CD}}\right] = g(d,\mathsf{E_b^{(m)}})
\end{array}
\right. 
\label{eq:E_CX} 
\end{align}
where the function $g$ is given by:
\begin{align}
g(k,\mathsf{E}) &= \exp \left( \mu_k + \frac{\sigma_k^2}{2} \right)
\frac{\Phi \left(- \sigma_k + \frac{\ln \left(\mathsf{E} \right) - \mu_k}{\sigma_k}  \right)}
{\Phi \left(\frac{\ln \left(\mathsf{E} \right) - \mu_k}{\sigma_k} \right)} .
\label{eq:g}
\end{align}
On the contrary, the conditional expectations  $ \mathbb{E}\left[E_d \, \vert \, \mathcal{C}_{\text{ED}}\right]$ and $ \mathbb{E}\left[E_s + E_r \, \vert \, \mathcal{C}_{\text{ER}}\right]$ for energy efficiency can only be expressed in an integral form. We thus propose to bound both of them, as done for the probabilities $\mathbb{P}_{\text{ER}}$ and $\mathbb{P}_{\text{ED}}$ in the previous section.

\begin{lemma}
When RTx is more energy-efficient than DTx, the average consumed energy $\mathbb{P}_{\text{ER}} \mathbb{E} \left[E_{b+r} \, \vert \, \mathcal{C}_{\text{ER}}\right] $ in log-normal shadowing environments is lower-bounded by $\mathsf{E}_{\text{low}}^{\text{(ER)}}$, where
\begin{align*}
\mathsf{E}_{\text{low}}^{\text{(ER)}} = 
\exp \left( \mu_{b+r} + \frac{\sigma_{b+r}^2}{2} \right)\mathbb{P}_{\text{low}}^{(1,\text{ER})}
+ g(b+r,\mathsf{E_R^{(m)}}) \mathbb{P}_{\text{low}}^{(2)}.
\end{align*}
Here, $\mathbb{P}_{\text{low}}^{(1,\text{ER})}$ is computed similarly to $\mathbb{P}_{\text{low}}^{(1)}$ but considering the scaled distribution $\exp(\sigma_{b+r}^2)E_{b+r}$ rather than $E_{b+r}$. We recall that $\mathbb{P}_{\text{low}}^{(1)}$ and $\mathbb{P}_{\text{low}}^{(2)}$ are given by Eq. \eqref{eq:Prob_low1} and \eqref{eq:Prob_low2} respectively.
\label{lemma:E_ER}
\end{lemma} 
\begin{proof}
See Appendix \ref{app:E_ER}.
\end{proof}

\begin{lemma}
When DTx is more energy-efficient than RTx, the average energy consumption is equal to $\mathbb{P}_{\text{ED}} \mathbb{E} \left[E_{d} \, \vert \, \mathcal{C}_{\text{ED}}\right] $ and is upper-bounded by Eq. \eqref{eq: E_ED} at the top of page.
\addtocounter{equation}{1}
$\mathbb{P}_{\text{low}}^{(1,\text{ED})}$ is computed similarly to $\mathbb{P}_{\text{low}}^{(1)}$ but with the scaled distribution $\exp(\sigma_{d}^2)E_{d}$ instead of $E_d$.
\label{lemma:E_ED}
\end{lemma}%\vspace*{-20pt}
\begin{proof}
The proof is similar to that of Lemma \ref{lemma:E_ER}.
\end{proof}

We now apply Lemmas \ref{lemma:E_ER} and \ref{lemma:E_ED} to deduce a model for the RF energy consumption.
\begin{model}[\textbf{EEA}]
We model the Energy Efficiency Area (EEA) of a relay-aided network in shadowing environment by the pair $(\widehat{\mathcal{A}_{\mathsf{E}}},\mathsf{E}_T)$. The average transmit energy to any mobile user $M \left(x,y\right)$ within $\widehat{\mathcal{A}_{\mathsf{E}}}$ can be assumed to be below $\mathsf{E}_T$, i.e.
\begin{align*}
\widehat{\mathcal{A}_{\mathsf{E}}} = \left\lbrace
M \left(x,y\right) \quad \text{s.t.} \quad \mathbb{E} \left[\widehat{E^{\mathsf{(RF)}}}\right] \leq \mathsf{E}_T
\right\rbrace,
\end{align*}
where $\mathbb{E} \left[\widehat{E^{\mathsf{(RF)}}}\right]$ is given by Eq. \eqref{eq:average_energy} but with the bounds $\mathsf{E}_{\text{low}}^{\text{(ER)}}$ and $\mathsf{E}_{\text{up}}^{\text{(ED)}}$. % given in  Lemmas \ref{lemma:E_ER} and \ref{lemma:E_ED}.
\label{model:area_energy_model}
\end{model}
Simulations show that $\mathbb{E} \left[\widehat{E^{\mathsf{(RF)}}}\right]$ can be considered as a tight upper-bound for $\mathbb{E} \left[E^{\mathsf{(RF)}}\right]$ and any user located with the model area $\widehat{\mathcal{A}_{\mathsf{E}}}$ can be assumed to be within ${\mathcal{A}_{\mathsf{E}}}$ as well.
The proposed model for the EEA can be extended to account for the circuitry energy consumption by using the same variable replacement as for the REA, given in Eq. \eqref{eq:model_extension_subs}.

\section{A new framework for joint analysis of relay-generated ICI and energy}
\label{sec:IEA}

The assumption that all interferers are transmitting at maximum power is not realistic within an energy-efficiency context, and we propose to characterize the impact of a RS on the interference imposed on a neighboring user, given the  \textit{actual} relay energy consumption.
Such an analysis provides helpful support for interference management techniques as well. 
Indeed, techniques to mitigate interference from neighboring cells, as considered for next generation cellular OFDMA systems \cite{LTE_Fundamentals, interfence_management}, require a fine understanding of the interference profile over time and frequency. This is particularly true in a network where both DTx and RTx occur simultaneously.

\subsection{Approximation of the relay-generated interference}

Based on the previous analysis for probability of energy-efficient relaying and overall energy consumption, we compute the average interference received at a mobile user located in another cell, at distance $d_I$ from the interfering relay, as illustrated in Figure \ref{fig:cell_system}.
To isolate the impact of a relay station, we do not consider other source of interference. This assumption is fair given that relays are usually equipped with omnidirectional antennas, as opposed to base stations.

We denote $M_I(x_I,y_I)$ as a given neighboring user, $\gamma_{I} = K_I d_I^{\alpha_I}$ as the path-loss between the relay station and this user, and $s_I$ as the corresponding shadowing coefficient. The average interference received at $M_I(x_I,y_I)$ is expressed as

\begin{align}
\hspace*{-15pt} I(x_I,y_I) 
= \mathbb{E}_{s_{I}}\left[\mathbb{E} \left[E_r^{\mathsf{(RF)}}\right]\frac{s_{I}}{\gamma_{I}}\right]
= \mathbb{E} \left[E_r^{\mathsf{(RF)}}\right]\frac{\exp\left(\sigma_I^2 / 2\right)}{K_I d_I^{\alpha_I}}
\label{eq:interference_mxy}
\end{align}
where $ \mathbb{E} \left[E_r^{\mathsf{(RF)}}\right]$ is the energy radiated by the relay, averaged over $s_d$, $s_s$ and $s_r$. It is given by
\begin{align}
\hspace*{-15pt} \mathbb{E} \left[E_r^{\mathsf{(RF)}}\right] = & \;
	\mathbb{P}_{\text{CR}} \, \mathbb{E}\left[E_r^{\mathsf{(RF)}} \, \vert \,  \mathcal{C}_{\text{CR}}\right]
	+ \mathbb{P}_{\text{ER}} \, \mathbb{E}\left[E_r^{\mathsf{(RF)}} \, \vert \, \mathcal{C}_{\text{ER}}\right]
	\label{eq:average_energy_relay}
	.
	%= \; \mathbb{P}_{\text{CR}} g(r,\mathsf{E_R^{(m)}}) + \mathbb{P}_{\text{ER}} \, \mathbb{E}\left[E_r \, \vert \, \mathcal{C}_{\text{ER}}\right]
\end{align}
First, the energy radiated by the relay when DTx is not feasible, due to fading, is computed as:
\begin{align}
	\mathbb{P}_{\text{CR}} \, \mathbb{E}\left[E_r^{\mathsf{(RF)}} \, \vert \,  \mathcal{C}_{\text{CR}}\right]
	= \mathbb{P}_{\text{CR}} g(r,\mathsf{E_R^{(m)}})
\end{align}
with $g$ given in Eq. \eqref{eq:g}. 
Second, using a proof similar to Lemma \ref{lemma:E_ER}, the energy radiated $\mathbb{P}_{\text{ER}} \mathbb{E} \left[E_{r}^{\mathsf{(RF)}} \, \vert \, \mathcal{C}_{\text{ER}}\right]$ by the RS when RTx is more energy-efficient than DTx, is lower-bounded by $\mathsf{E}_{\text{low}}^{\text{(I)}} $ with 
\begin{align}
\mathsf{E}_{\text{low}}^{\text{(I)}} 
 = 
\exp \left( \mu_{r} + \frac{\sigma_{r}^2}{2} \right) \left(\mathbb{P}_{\text{low}}^{(1,\text{I})}
+ \mathbb{P}_{\text{low}}^{(2,\text{I})} \right),
\label{eq:I_ER}
\end{align}
where $\mathbb{P}_{\text{low}}^{(1,\text{I})}$ and $\mathbb{P}_{\text{low}}^{(2,\text{I})}$ are computed similarly to $\mathbb{P}_{\text{low}}^{(1)}$ and $\mathbb{P}_{\text{low}}^{(2)}$ in Eq. \eqref{eq:Prob_low1} and \eqref{eq:Prob_low2}, but considering the scaled distribution $\exp(\sigma_{r}^2)E_{r}$ instead of $E_{r}$.
%\label{lemma:I_ER}
%\end{lemma}

\begin{model}[\textbf{Approximation for $I(x_I,y_I)$}]
The relay-generated interference at a given neighboring user $M_I(x_I,y_I)$, located at a distance $d_I$ from the relay station, is lower-bounded by
\begin{align*}
\widehat{I}(x_I,y_I) &=
\mathbb{E} \left[\widehat{E_r^{\mathsf{(RF)}}}\right]\frac{\exp\left(\sigma_I^2 / 2\right)}{K_I d_I^{\alpha_I}} 
\\&=
\left( \mathbb{P}_{\text{CR}}g(r,\mathsf{E_R^{(m)}}) + \mathsf{E}_{\text{low}}^{\text{(I)}}\right)\frac{\exp\left(\sigma_I^2 / 2\right)}{K_I d_I^{\alpha_I}},
%\leq I(x_I,y_I),
\end{align*}
where $g$ is given by Eq. \eqref{eq:g} and $\mathsf{E}_{\text{low}}^{\text{(I)}}$ by Eq. \eqref{eq:I_ER}.
\label{model:area_interference_model}
\end{model}

One can argue that upper-bounding the interference $\widehat{I}(x_I,y_I)$ would be more suitable for performance analysis. However, we highlight that the proposed lower-bound is tight, as shown in Section \ref{sec:validation}. It is thereby accurate enough to investigate interference-aware relay deployment. The proposed model for the ICI can be extended to account for the circuitry energy consumption by using the variable replacement of Eq. \eqref{eq:model_extension_subs}, similarly to REA and EEA.

\subsection{A new metric for analyzing energy and interference}

A relay station can provide significant \emph{energy gain} and coverage extension for the cell it serves. But, at the same time, it is an additional source of interference, implying that neighboring cells experience an \emph{energy loss} to maintain the same data rate for their own users.
In consequence, a relay deployment is efficient if the achieved energy gain, referred as $\upsilon_{\text{Gain}}$, is higher than the resulted energy loss, referred as $\upsilon_{\text{Loss}}$. We propose to use their ratio as a metric to jointly capture the aspects of energy and interference.

To evaluate the energy gain $\upsilon_{\text{Gain}}$, we consider a user $M(x,y)$ served by BS$_1$. We compare the energy $\mathsf{E_{1}^{(N_r=0)}}$ consumed to send data to this user when BS$_1$ is not supported by relay stations ($N_r=0$) and the energy $\mathsf{E_{1}^{(N_r)}}$ consumed when BS$_1$ is supported by $N_r$ relay stations. We have:
\begin{align*}
\left \lbrace \begin{array}{r l}
 \mathsf{E_{1}^{(N_r=0)}} &\hspace*{-5pt} \left[ x,y\right] =
 \eta_B E^{\mathsf{(RF)}} \left[ x,y\right]
+ \left( \mathsf{E_B^{(Tx)}} + \mathsf{E_U^{(Rx)}} \right)
+  \mathsf{E_B^{(idle)}} 
 \\
 \mathsf{E_{1}^{(N_r)}} & \hspace*{-5pt} \left[ x,y\right] =
 \eta_B E_B^{\mathsf{(RF)}} \left[ x,y\right]
+ \eta_R E_R^{\mathsf{(RF)}} \left[ x,y\right]
\\ & \hspace*{-5pt}
+ \left( \mathsf{E_B^{(Tx)}} + \mathsf{E_U^{(Rx)}} + \mathsf{E^{(dsp)}}\right)
+ \left( \mathsf{E_B^{(idle)}} + N_r \mathsf{E_R^{(idle)}} \right).
 \end{array}
 \right.
\end{align*}
Here, the various energy offsets accounted in $\mathsf{E_{idle}}$ are described in Section \ref{sec:coding_scheme}. We recall that we can focus only on Sector 1 since the three base stations BS$_1$, BS$_2$ and BS$_3$ within the considered cell are orthogonal and do use the same resource (in time and frequency).

Similarly, to evaluate the energy loss $\upsilon_{\text{Loss}}$, we consider a user $M(x,y)$, which is located in a neighboring cell $i\neq1$ and performs DTx. We compare the energy $\mathsf{E_{i}^{(N_r=0)}}$ consumed by the neighboring BS$_i$ to send data to this user when BS$_1$ is not supported by relay stations ($N_r=0$) and the energy $\mathsf{E_{i}^{(N_r)}}$, consumed to maintain the same rate, when BS$_1$ is supported by $N_r$ relay stations generating interference.
To isolate the impact of the $N_r$ relays, we assume an ideal network, without any other source of interference. Denoting $I\left( x,y\right)$ the interference received at $M(x,y)$ as defined in Eq. \eqref{eq:interference_mxy} of previous interference analysis,  we have $\forall i \neq 1$
\begin{align*}
\left \lbrace \begin{array}{r l}
 \mathsf{E_{i}^{(N_r=0)}} \left[ x,y\right] & \hspace*{-5pt}=  \eta_B E^{\mathsf{(RF)}} \left[ x,y\right]
+ \left( \mathsf{E_B^{(Tx)}} + \mathsf{E_U^{(Rx)}} \right)
+ \mathsf{E_B^{(idle)}}
 \\
\mathsf{E_{i}^{(N_r)}} \left[ x,y\right] & \hspace*{-5pt}=   \eta_B E^{\mathsf{(RF)}} \left[ x,y\right] \left(1 + \frac{2 I\left( x,y\right)}{N} \right)
\\
&+ \left( \mathsf{E_B^{(Tx)}} + \mathsf{E_U^{(Rx)}} \right)
+ \mathsf{E_B^{(idle)}} .
 \end{array}
 \right.
 \end{align*}

\textit{Remark:} In addition to energy gain, relay stations also provide coverage extension. To account for such extension and compute $\upsilon_{\text{Gain}}$ and $\upsilon_{\text{Loss}}$, we do not consider a power constraint for DTx. 

\begin{definition}
To capture both the energy and interference aspects, we define the ratio $\Gamma$ as
\begin{align}
\Gamma = \frac{\upsilon_{\text{Gain}}}{\upsilon_{\text{Loss}}} \quad 
 \text{with} \;
\left \lbrace
\begin{array}{l}
\upsilon_{\text{Gain}} =  \mathbb{E}\left[\frac{\mathsf{E_{1}^{(N_r=0)}}-\mathsf{E_{1}^{(N_r)}}}{\mathsf{E_{1}^{(N_r=0)}}} \right]
\\
\upsilon_{\text{Loss}} = \mathbb{E}\left[\frac{\underset {i \neq 1} {\sum} \mathsf{E_{i}^{(N_r=0)}}- \underset {i \neq 1} {\sum} \mathsf{E_{i}^{(N_r)}}}{ \underset {i \neq 1} {\sum} \mathsf{E_{i}^{(N_r=0)}}} \right]
\end{array}
\right.
\end{align}
\end{definition}
Here, $\upsilon_{\text{Gain}}$ is averaged over all users served by BS$_1$ and $\upsilon_{\text{Loss}}$ is averaged over all users located in the neighboring cells 2 to 7, as depicted in Figure \ref{fig:cell_system}.
If $\Gamma > 1$, the considered relay configuration is efficient, if $0< \Gamma < 1$, the relay stations result in more energy loss for neighboring cells than they actually provide energy gain in their own cell. If $\Gamma < 0$, relaying does not provide any energy gain, due to the circuitry consumption.

%, such that 
%consider a larger outage requirement for direct transmissions and $\forall i$, 
%the transmit energy part in $\mathsf{E_{i}^{(N_r=0)}}$ and  $\mathsf{E_{i}^{(N_r)}}$ is upper-bounded by $\mathsf{E_B^{(m)}}=2W$. Above such limit, $\upsilon_{\text{Gain}}$ largely exceeds $\upsilon_{\text{Loss}}$ in any case.
%
%\textit{Remark 2:}
%Also note that the $\Gamma$-metric does not explicitly use the EEA and the IEA defined in earlier sections, but only the approximations of the energy consumption $\mathbb{E}\left[\widehat{E^{\mathsf{(RF)}}}\right]$ in the computation of $\mathsf{E_{1}^{(N_r)}}$ and of the interference $\widehat{I}(x,y)$ in $\mathsf{E_{i}^{(N_r)}}$.

\begin{figure*}[th!]
\centering
\subfigure[Relay Efficiency Area]{
\begin{tikzpicture}
 \begin{axis}[
   width  = 0.33*\textwidth,
        height = 5cm,
        ymin=0, ymax=9,
        y label style={at={(axis description cs:0.2,0.5)},rotate=0,anchor=south},
		xtick={0.3,0.4,0.5,0.6,0.7,0.8},
        xlabel={$ \mathbb{P}_T$},
        ylabel={$\zeta_{\text{R}}$ in \%},
        legend style={at={(0.5,-0.2)},
        anchor=north,legend columns=-1},
        ybar,
        ymajorgrids = true,
         every axis/.append style={font=\footnotesize},  
      bar width=3pt,
         ]
         \addplot coordinates 
        {(0.30,0.37) (0.40,0.38) (0.50,0.365) (0.60,0.54) (0.70,0.96) (0.80,2.85)};
        
        \addplot coordinates 
        {(0.30,0.355) (0.40,0.645) (0.50,0.9583) (0.60,1.8017) (0.70,3.6567) (0.80,8.4)
        };
        
        \addplot coordinates 
        {(0.30,0.55) (0.40,0.9483) (0.50,1.54) (0.60,2.82) (0.70,5.82)};
        
    \end{axis}
    
\end{tikzpicture}
}
%%%%%%%%%%%%%%
\subfigure[Energy Efficiency Area]{
\begin{tikzpicture}
 \begin{axis}[
  width  = 0.33*\textwidth,
        height = 5cm,
        ymin=0, ymax=1.6,
        y label style={at={(axis description cs:0.15,0.5)},rotate=0,anchor=south},
		xtick={150,250,350,450,550,650},
        xlabel={$ \mathsf{E}_T$ in mJ},
        ylabel={$\zeta_{\text{E}}$ in \%},
        legend style={at={(0.5,-0.2)},
        anchor=north,legend columns=-1},
        ybar,
        ymajorgrids = true,
         every axis/.append style={font=\footnotesize},  
      bar width=3pt,
          ]
         \addplot coordinates 
        {(150,0.16) (250,0.34) (350,0.48) (450,0.65) (550,0.85) (650,0.998)};
        
        \addplot coordinates 
       {(150,0.49) (250,0.873) (350,0.7567) (450,0.525) (550,0.50166) (650,0.65833)};
         \addplot coordinates 
       {(150,0.13) (250,1.49) (350,1.246) (450,0.85166) (550,0.655) (650,0.4733)};
         
       \end{axis}
\end{tikzpicture}
}
%%%%%%%%%%%%%%%%%%%%%$\mathbb{E}\left[\widehat{E_r^{j}}\right]$
\subfigure[Approximation of $\mathbb{E} \left( \widehat{E_r^{(RF)}} \right)$]{
\begin{tikzpicture}
 \begin{axis}[
  width  = 0.33*\textwidth,
        height = 5cm,
        ymin=0, ymax=5.1,
        y label style={at={(axis description cs:0.2,0.5)},rotate=0,anchor=south},
		xtick={20,40,60,80,100},
		ytick={1,3,5},
		xlabel={$ \mathsf{E}_{T,r}$ in mJ},
        ylabel={$\zeta_{\text{I}}$ in \%},
        legend style={at={(0.5,-0.2)},
        anchor=north,legend columns=-1},
        ybar,
        ymajorgrids = true,
         every axis/.append style={font=\footnotesize},  
      bar width=3pt,
         ]
         
 \addplot coordinates 
        {(20,1.522) (40,1.453) (60,1.645) (80,2.44) (100,4.03)
        % (120,6.36)
        };
        
        \addplot coordinates 
        {(20,1.14) (40,1.308) (60,1.86) (80,2.6716) (100,4.25) 
        %(120,6.195)
        };
        
        \addplot coordinates 
        {(20,1.28) (40,1.54) (60,2.015) (80,2.83) (100,4.415)
        % (120,6.108)
        };         
         
%         \addplot coordinates 
%        {(20,1.30) (40,1.57) (60,2.29) (80,3.26) (100,5.015) (120,7.095)};
%        
%        \addplot coordinates 
%        {(20,1.36) (40,1.86) (60,2.56) (80,3.47) (100,5.24) (120,6.86)};
%        
%        \addplot coordinates 
%        {(20,1.596) (40,2.075) (60,2.695) (80,3.612) (100,5.3766) (120,6.685)};
        
    \end{axis}
\end{tikzpicture}
}

\centering \fbox{ \footnotesize
\tikz{ \draw [fill=blue!30!white,  draw=blue](0,0) rectangle (10pt,6pt); }
$\mathsf{E_{2Hop}^{(dsp)}}$= 0mJ \quad
\tikz{ \draw[fill=red!30!white,  draw=red] (0,0) rectangle (10pt,6pt); }
$\mathsf{E_{2Hop}^{(dsp)}}$= 50mJ \quad
\tikz{ \draw [fill=brown!30!white,  draw=brown](0,0) rectangle (10pt,6pt); }
$\mathsf{E_{2Hop}^{(dsp)}}$= 100mJ \quad 
}
\caption{Validation of the proposed models: error ratio $\zeta$ for various simulation settings}
\label{fig:validation}
\end{figure*}
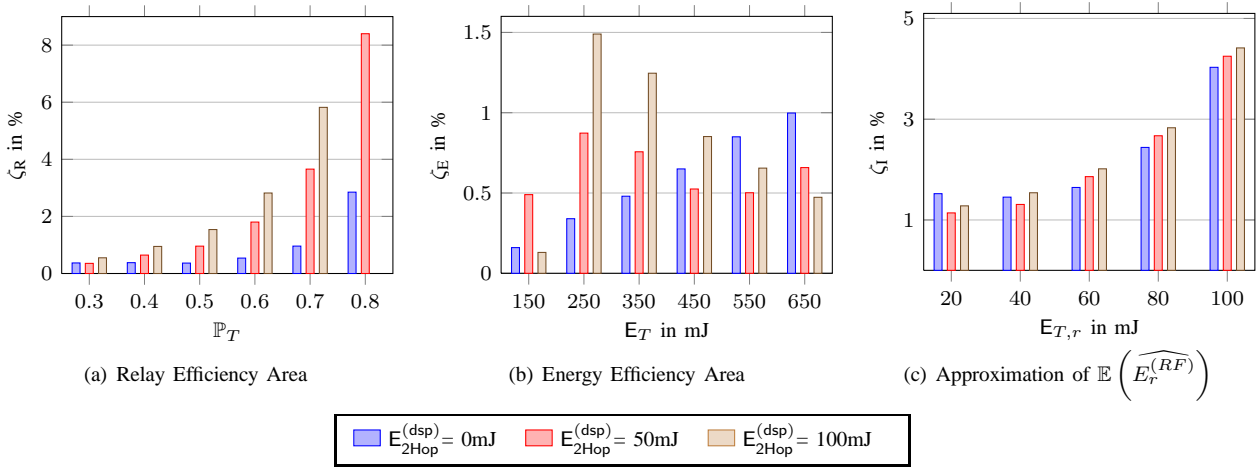

%%%%%%%%%%%%%%%%%%%%%%%%%%%%%%%%%%%%%%%%%%%%%%
\section{Performance analysis for Energy- and ICI-efficient Relay Deployment}
\label{sec:performance}

In this section, we first validate the proposed models for relaying probability, energy consumption and interference. Then, we jointly analyze the energy consumption and the generated ICI using two-hop relaying. In the last subsection, the impact of the relay coding scheme on the network performance is explored. 
If not specified, we consider the simulation parameters of Table \ref{sim_param}, taken from \cite{correia2010, Auer2011,andreev2012,book_green_network}. % and  \cite[Chapter~16]{book_green_network}.
For the channel gains,
%we consider the path-loss model presented in the WINNER II project \cite{winner}. In this,
the direct link $h_d$ and interference link $h_I$ are modelled by scenario C2 of the WINNER II project \cite{winner}, the RS-to-user link $h_r$ by scenario B1, the BS-to-RS link $h_b$ by B5c. We recall that we consider normalized transmissions of unitary length, setting up a direct relation between energy and power.

\begin{table}
\centering

\begin{tabular}{|c|c|c|c|c|}
\hline
\multirow{5}{*}{Energy} 
& $\mathsf{E_B^{(m)}}$ & 1J 
& $\mathsf{E_R^{(m)}}$ & 500mJ \\ 
& $\mathsf{E_B^{(idle)}}$ & 25mJ 
& $\mathsf{E_R^{(idle)}}$ & 10mJ \\
& $\mathsf{E^{(dsp)}_{2Hop}}$ & 0-50mJ 
& $\mathsf{E^{(dsp+)}_{pdf}}$ & 0-50mJ \\
& $\eta_B$ & 2.66 
& $\eta_R$ & 3.1 \\
& \multicolumn{2}{|c|}{$\mathsf{E_B^{(Tx)}} + \mathsf{E_U^{(Rx)}} $} & \multicolumn{2}{|c|}{90mJ}\\
\hline
\end{tabular}

\vspace{5pt}
\begin{tabular}{|c|c|c|c|c|c|}
\hline
\multirow{6}{*}{Channel} 
& $N$ & -93dBm &
\multirow{6}{*}{Others}
& $\mathbb{P}_{\text{out}}$ & 0.02 \\ 
& $\sigma_d$ & 6dB &
& $\mathcal{R}$ & 3bit/ch.use \\
& $\sigma_b$ & 3dB &
& $f_c$ & 2.6GHz \\ 
& $\sigma_r$ & 4dB &
& $\mathsf{H_b}$ & 30m \\ 
& $\sigma_I$ & 6dB  &
& $\mathsf{H_r}$ & 20m \\ 
& \multicolumn{2}{|c|}{Normalized Tx (1s)} &
& $\mathsf{H_u}$ & 1.5m \\ 
\hline
\end{tabular}
\caption{Simulation parameters }
\label{sim_param}
\end{table}

\subsection{Models validation}
\label{sec:validation}

For validation of the proposed model for REA, we account for all users $M(x,y)$ located in the simulated $\left(\mathcal{A}_{\text{R}},\mathbb{P}_T \right)$ but not declared in $\left(\widehat{\mathcal{A}_{\text{R}}},\mathbb{P}_T \right)$, for some given threshold $\mathbb{P}_T$ (or reversely, $M(x,y)$ is declared inside while it is actually outside). 
This means that, for such user, the effective relaying probability, obtained by Monte-Carlo simulations, is s.t. $\mathbb{P}_{\text{RTx}} \geq \mathbb{P}_T$ but the proposed lower-bound gives $ \mathbb{P}_{\text{low}} \left(x,y\right) + \mathbb{P}_{\text{CR}}\left(x,y\right) \leq \mathbb{P}_T$ (or reversely).  
We define the error ratio $\zeta_{\text{R}}$ as the proportion of such erroneous users, i.e.
\begin{align*}
&\zeta_{\text{R}} = \frac{\iint \mathds{1}_{\mathcal{E}_{\text{R}}} dx dy}{\iint \mathds{1}_{\mathcal{A}_{\text{R}}} dx dy}
\end{align*}
\begin{align*}
\text{with} \quad & \mathcal{E}_{\text{R}} = \left\lbrace
M \left(x,y\right) \in \mathcal{A}_{\text{R}}
\; \cap \;
 M \left(x,y\right) \notin \widehat{\mathcal{A}_{\text{R}}}
\right\rbrace
\\ & \quad
 \cup \;
\left\lbrace
M \left(x,y\right) \notin \mathcal{A}_{\text{R}}
\; \cap \;
 M \left(x,y\right) \in \widehat{\mathcal{A}_{\text{R}}}
\right\rbrace 
\end{align*}

Similarly, the error ratio $\zeta_{\text{E}}$ for the EEA refers to the proportion of erroneous users, for which $\mathbb{E} \left[E^{\mathsf{(RF)}}\right] \leq \mathsf{E}_T$ and $\mathbb{E} \left[\widehat{E^{\mathsf{(RF)}}}\right] \geq \mathsf{E}_T$ (or reversely), for some given threshold $\mathsf{E}_T$. 
For the interference analysis, we focus on the approximation of the average energy radiated by the relay and define the error ratio $\zeta_{\text{I}}$ as the proportion of users for which  $\mathbb{E} \left[E_r^{\mathsf{(RF)}}\right] \leq \mathsf{E}_{T,r}$ and $\mathbb{E} \left[\widehat{E_r^{\mathsf{(RF)}}}\right] \geq \mathsf{E}_{T,r}$ (or reversely), for some given threshold $\mathsf{E}_{T,r}$.

We plot in Figure \ref{fig:validation}(a) (resp. b and c) the error ratio $\zeta_{\mathsf{R}}$ (resp. $\zeta_{\mathsf{E}}$ and $\zeta_{\mathsf{I}}$) obtained for a wide range of $\mathbb{P}_T$ (resp. $\mathsf{E}_T$ and $\mathsf{E}_{T,r}$) and several $\mathsf{E^{(dsp)}}$. More precisely, the plotted $\zeta_{\text{X}}$ stands for the ratio averaged over various RS-to-BS distances ($\mathsf{D}_b \in \left[600,1000\right]$m), various user rates ($\mathcal{R} \in \left[2,4\right]$bits/ch. use) and the two outdoors propagation environments described in \cite[Appendix A]{Journal2}.
%It accounts for the errors induced by the use of a lower-bound for $\mathbb{P}_{\text{RTx}} \left(x,y\right)$ and by the Fenton-Wilkinson approximation.
For the purpose of validation, we consider a wide cell coverage by fixing the outage requirement $\mathbb{P}_\text{out}$ to 0.1, which is very large for encoded data.

From Figure \ref{fig:validation}(a), we observe that the error ratio for the REA does not exceed 3\% when the circuitry consumption is not considered ($\mathsf{E^{(dsp)}}=0$). Although the model for $\mathsf{E^{(dsp)}}>0$ is less accurate, such error increase does not impact at all the proposed model for energy consumption, as illustrated in Figure \ref{fig:validation}(b). 
Indeed, when an error occurs and a user $M(x,y)$ is wrongly declared in $\widehat{\mathcal{A}_{\mathsf{R}}}$ while it is not (or reversely), we have $E_{b} + E_{r} \simeq E_{d}$.
For the EEA, the error ration does not exceed 1.5\% and, for the interference approximation, it is below 5\%, as plotted in Figure \ref{fig:validation}(c). Furthermore, error are mostly located at cell edge and, by considering restricted cell coverage ($\mathbb{P}_\text{out}=0.02$), as for the rest of this paper, the error ratio for the ICI falls under 2.5\%.

The proposed models for the REA, EEA and ICI thus provide very efficient frameworks for performance analysis with regards to both accuracy and savings in the simulation time. Non-model based simulations were also performed taking 50 000 samples for the channel gains of each link and required several hours for a single relay configuration. By comparison, model-based simulations were completed in less than 3 seconds. Subsequently, we will refer to "simulations" for "model-based simulations".

\subsection{On the minimal energy consumption per unit area}

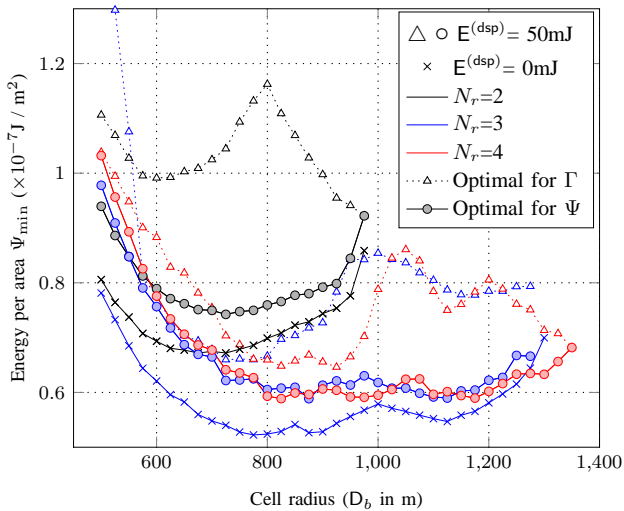
\begin{figure} 
\centering
\centering \resizebox{0.95\columnwidth}{!}{%

\begin{tikzpicture}[scale=4,cap=round,>=latex]

 \pgfplotsset{
    grid style = {
      dash pattern = on 0.05mm off 1mm,
      line cap = round,
      black,
      line width = 0.5pt
    }
  }

  \begin{axis}[%
  	ymin=0.5, ymax=1.3,
  	xmin=450, xmax=1400,
    xlabel=Cell radius ($\mathsf{D}_b$ in m),%
    ylabel=Energy per area $\Psi_{\min}$ ($\times 10^{-7}$J / m$^2$),%
    y label style={at={(axis description cs:0.09,0.5)},rotate=0,anchor=south},
%   width=\textwidth/1.25,
	height=8cm,
%	width=5cm,
	%axis equal,
	every axis/.append style={font=\footnotesize},  
    grid=major,%
    legend style={at={(axis cs:1390,1.295)},anchor=north east, nodes=right, font=\small},%
    %legend pos={south west},%
    mark size=1.8pt]

%%%%%%%%%%%%%%%%% LEGEND E_DELTA	
	
\addplot[only marks, color=black,mark=o/triangle] table[x=coverage ,y=eea ,col sep=semicolon] {data_Energy_Nr2_050mW_new.txt};
	\addlegendentry{\hspace*{-22pt} $\triangle$ {\large $\circ$} $\mathsf{E^{(dsp)}}$= 50mJ}  
	    	
    \addplot[only marks, color=black, mark=x,mark size=2pt] table[x=coverage ,y=eea ,col sep=semicolon] {data_Energy_Nr2_000mW_new.txt};
	\addlegendentry{$\mathsf{E^{(dsp)}}$= 0mJ}  
	
%
%%%%%%%%%%%%%%%% LEGEND Nr
%
    \addplot[solid, color=black] table[x=coverage ,y=eea ,col sep=semicolon] {data_Energy_Nr2_050mW_new.txt};
	\addlegendentry{$N_r$=2}  
\addplot[solid, color=blue] table[x=coverage ,y=eea ,col sep=semicolon] {data_Energy_Nr3_050mW_combined.txt};
	\addlegendentry{$N_r$=3}  
	
	\addplot[solid, color=red] table[x=coverage ,y=eea ,col sep=semicolon]
	{data_Energy_Nr4_050mW_combined.txt};
	\addlegendentry{$N_r$=4}

%%%%%%%%%%%%%%% GAMMA

\addplot[dotted, color=black,mark=triangle*,mark options={solid,fill=white},mark size=1.8pt]
table[x=coverage ,y=GAMMAeea ,col sep=semicolon]
{data_Energy_Nr2_050mW_new.txt};
\addlegendentry{Optimal for $\Gamma$}
%\addlegendentry{Optimal for $\mathbb{E}$}

\addplot[dotted, color=blue,mark=triangle*,mark options={solid,fill=white},mark size=1.8pt,forget plot]
table[x=coverage ,y=GAMMAeea ,col sep=semicolon]
{data_Energy_Nr3_050mW_combined.txt};

\addplot[dotted, color=red,mark=triangle*,mark options={solid,fill=white},mark size=1.8pt,forget plot]
table[x=coverage ,y=GAMMAeea ,col sep=semicolon]
{data_Energy_Nr4_050mW_combined.txt};

%%%%%%%%%%%%%%%%% PSI E_Delta = 0

\addplot[solid, color=blue,mark=x,mark options={solid,fill=blue!30!white},mark size=2pt,forget plot]
table[x=coverage ,y=eea ,col sep=semicolon]
{data_Energy_Nr3_000mW_corrected.txt};

\addplot[solid, color=black,mark=x,mark options={solid,fill=black!30!white},mark size=2pt,forget plot]
table[x=coverage ,y=eea ,col sep=semicolon]
{data_Energy_Nr2_000mW_new.txt};

%%%%%%%%%%%%%%%%% PSI E_Delta = 50

\addplot[solid, color=black,mark=*,mark options={solid,fill=black!30!white}]
table[x=coverage ,y=eea ,col sep=semicolon]
{data_Energy_Nr2_050mW_new.txt};
\addlegendentry{Optimal for $\Psi$}

\addplot[solid, color=blue,mark=*,mark options={solid,fill=blue!30!white},forget plot]
table[x=coverage ,y=eea ,col sep=semicolon]
{data_Energy_Nr3_050mW_combined.txt};

\addplot[solid, color=red,mark=*,mark options={solid,fill=red!30!white},forget plot]
table[x=coverage ,y=eea ,col sep=semicolon]
{data_Energy_Nr4_050mW_combined.txt};

\end{axis}

\end{tikzpicture}
}
\caption{Minimal energy consumption per unit area $\Psi_{\min}$}
\label{fig:TO_Energy_Coverage}
\end{figure}

As a first step, we do not consider the impact of relays in terms of interference and analyze the relay performance regarding the EEA only, in the shadowing model used in the WINNER project.
The Joule-per-bit metric has been widely used for energy efficiency analysis. Yet, in practice, a large part of the network is primarily providing coverage and does not operate at full load, even at peak traffic hours. Due to the energy dissipated in circuitry to maintain the network operational, the energy efficiency can be particularly poor under low-traffic loads and restricted coverage \cite{correia2010}. 
To capture the aspect of the cell coverage, we consider the maximal energy consumption $\mathsf{E}_{\max}$ that is required to send data at rate $\mathcal{R}$ to any user located within the cell sector, i.e. $\mathcal{A}_{\text{sector}} \in \left(\mathcal{A}_{\mathsf{E}},\mathsf{E}_{\max}\right)$, and divide it by the sector area $\mathcal{A}_{\text{sector}}$. It is expressed in Joule-per-square-meter and denoted as $\Psi$:
\begin{align}
\Psi = \frac{\mathsf{E}_{\max} + \mathsf{E_{idle}}}{\mathcal{A}_{\text{sector}}}
\quad \text{where} \quad
\mathcal{A}_{\text{sector}} = \frac{\sqrt{3}}{2} \mathsf{D}_b^2.
\end{align}
Here, $\mathsf{E}_T$ accounts for the transmit energy $\mathbb{E} \left[E^{\mathsf{(RF)}}\right]$ as in Eq. \eqref{eq:average_energy}, the RF amplifier coefficients $\eta_R$ and $\eta_B$ and the energy $\mathsf{E^{(dsp)}}$ dissipated at the relay for decoding and re-encoding. $\mathsf{E_{idle}}$ refers to the other energy offsets described in Section \ref{sec:coding_scheme}, for the considered coding scheme. %In addition, we assume an outage requirement of $\mathbb{P}_\text{out}=0.02$.
In Figure \ref{fig:TO_Energy_Coverage}, we plot the minimal feasible energy per unit area $\Psi$ as a function of the cell radius $\mathsf{D}_b$, considering different values for $N_r$ and $\mathsf{E^{(dsp)}}$. To do so, for each considered set $\left( \mathsf{D}_b, N_r, \mathsf{E^{(dsp)}}\right)$, we find the location for the $N_r$ relays which minimizes $\Psi$. Note that proceeding this way would not have been possible in a reasonable time without using the proposed models.

\textit{Remark:} For comparison purpose, Figure \ref{fig:TO_Energy_Coverage} also plots the value for $\Psi$ achieved when the relay location maximizes the metric $\Gamma$ (namely "Optimal for $\Gamma$" in the figure). We will come back to this point in the following subsection.
% From this figure, we deduce the following results.

\begin{result}
The energy offset $\mathsf{E^{(dsp)}}$, consumed for decoding and re-encoding at the relay station has severe impact on the cell energy efficiency, except for very large cell radius, where the overall energy is dominated by the RF transmit consumption.
\end{result}
As an example, at $\mathsf{D}_b=$700m, increasing $\mathsf{E^{(dsp)}}$ from 0 to 50mW leads to a degradation in $\Psi_{\min}$ of 11\% with $N_r$=2 and of 22\% with $N_r$=3.

\begin{result}
Up to $N_r$=3, increasing the number of relays per cell allows significant energy reduction. Passed this limit, the gain provided by additional relays is not sufficient to compensate for the idle energy $N_r \mathsf{E_R^{(idle)}}$, dissipated to maintain the network operational.
\end{result}

With $N_r$=2, the minimal value $\Psi_{\min}$ is equal to 0.74e-7J/m$^2$, and adding one more RS (making Nr = 3) allows a gain of 28\%, where $\Psi_{\min}$ then reduces to  0.58e-7J/m$^2$.
However, there is no significant performance gain for increasing the number of RS from  $N_r$=3 to $N_r$=4.

We recall that we have considered as performance metric the maximal energy necessary to transmit data at a given rate to any user of the cell sector, as given by the EEA, rather than simply the average energy consumption of the cell. 
Simulations show that the relay configurations optimal for $\Psi$ and for the average do not match. When optimized for the average consumption, the relay configuration results in a severe energy increase at cell edge (from 10\% to 25\%), meaning that there exists a major performance gap between users of the cell center, with strong channel, and cell-edge users, with weak channel. Thus, the average optimization criteria does not provide fairness as does our proposed Energy Efficiency Area and related metric $\Psi$.

%. Such metric is relevant for performance analysis and encompasses fairness between the served users. Indeed, a relay station cannot be energy-efficient only for users located in a certain area of the cell, while showing unacceptably high energy consumption elsewhere.

\subsection{A new energy-interference trade-off on relay deployment}

\begin{figure} 
\centering
\centering \resizebox{0.95\columnwidth}{!}{%

\begin{tikzpicture}[scale=4,cap=round,>=latex]

 \pgfplotsset{
    grid style = {
      dash pattern = on 0.05mm off 1mm,
      line cap = round,
      black,
      line width = 0.5pt
    }
  }
%
%  \begin{axis}[%
%  ymin=-4.2, ymax=3.8,
%  ytick = {-4,-3,-2,-1,0,1,2,3},
%    xlabel=Coverage ($\mathsf{D}_b$ in m),%
%    ylabel=$\Gamma_{\max}$,%
%      y label style={at={(axis description cs:0.1,0.5)},rotate=0,anchor=south},
%%    width=\textwidth/1.25,
%	height=7cm,
%	%axis equal,
%	every axis/.append style={font=\footnotesize},  
%    grid=major,%
%    legend style={nodes=right, font=\footnotesize},%
%    legend pos={south east},
%    mark size=1.8pt]

  \begin{axis}[%
	ymin=-3.6, ymax=4.2,
	xmin=450, xmax=1400,
    ytick = {-4,-3,-2,-1,0,1,2,3},
	xlabel=Cell radius ($\mathsf{D}_b$ in m),%
    ylabel=Energy-to-interference ratio $\Gamma_{\max}$,%
    y label style={at={(axis description cs:0.09,0.5)},rotate=0,anchor=south},
%   width=\textwidth/1.25,
	height=8cm,
%	width=5cm,
	%axis equal,
	every axis/.append style={font=\footnotesize},  
    grid=major,%
      legend style={at={(axis cs:1390,-3.55)},anchor=south east, nodes=right, font=\small},%
    %legend pos={south west},%
   mark size=2pt]

% 

%%%%%%%%%%%%%%%% LEGEND E_DELTA	

\addplot[only marks, color=black,mark={triangle,o}] table[x=coverage ,y=gammaMax ,col sep=semicolon] {data_Gamma_cell_Nr2_050mW_new.txt};
	\addlegendentry{\hspace*{-22pt} $\triangle$ {\large $\circ$} $\mathsf{E^{(dsp)}}$= 50mJ}  
	    	
    \addplot[only marks, color=black, mark=x,mark size=2pt] table[x=coverage ,y=gammaMax ,col sep=semicolon] {data_Gamma_cell_Nr2_000mW_new.txt};
	\addlegendentry{$\mathsf{E^{(dsp)}}$= 0mJ}

%%%%%%%%%%%%%%%% LEGEND Nr

    \addplot[solid, color=black]  table[x=coverage ,y=gammaMax ,col sep=semicolon] {data_Gamma_cell_Nr2_050mW_new.txt};
	\addlegendentry{$N_r$=2}  
\addplot[solid, color=blue]  table[x=coverage ,y=gammaMax ,col sep=semicolon] {data_Gamma_cell_Nr3_050mW_combined.txt};
	\addlegendentry{$N_r$=3}  
 
 \addplot[solid, color=red]  table[x=coverage ,y=gammaMax ,col sep=semicolon] {data_Gamma_cell_Nr4_050mW_combined.txt};
	\addlegendentry{$N_r$=4} 
% 

%%%%%%%%%% gamma pour EEA min

\addplot[dotted, color=blue,mark=*,mark options={solid,fill=white},mark size=1.5pt,forget plot]
table[x=coverage ,y=gammaEEA ,col sep=semicolon]
{data_Gamma_cell_Nr3_050mW_combined.txt}; 	 
\addplot[dotted, color=red,mark=*,mark options={solid,fill=white},mark size=1.5pt,forget plot]
table[x=coverage ,y=gammaEEA ,col sep=semicolon]
{data_Gamma_cell_Nr4_050mW_combined.txt}; 	 

	  \addplot[dotted, color=black,mark=*,mark options={solid,fill=white},mark size=1.5pt]
table[x=coverage ,y=gammaEEA ,col sep=semicolon]
{data_Gamma_cell_Nr2_050mW_new.txt};
	\addlegendentry{Optimal for $\Psi$}

 %%%%%%% 0mW
 \addplot[solid, color=black,mark=x,mark size=2pt,mark options={solid,fill=black!30!white},forget plot]
table[x=coverage ,y=gammaMax ,col sep=semicolon]
{data_Gamma_cell_Nr2_000mW_new.txt}; 	

%%
%%%%%%% 50mW

\addplot[solid, color=black,mark=triangle*,mark options={solid,fill=black!30!white}]
table[x=coverage ,y=gammaMax ,col sep=semicolon]
{data_Gamma_cell_Nr2_050mW_new.txt};
	\addlegendentry{Optimal for $\Gamma$}

\addplot[solid, color=blue,mark=triangle*,mark options={solid,fill=blue!30!white},forget plot]
table[x=coverage ,y=gammaMax ,col sep=semicolon]
{data_Gamma_cell_Nr3_050mW_combined.txt}; 	
\addplot[solid, color=red,mark=triangle*,mark options={solid,fill=red!30!white},forget plot]
table[x=coverage ,y=gammaMax ,col sep=semicolon]
{data_Gamma_cell_Nr4_050mW_combined.txt}; 	  
	  
%%	
%\addplot[solid, color=Orange,mark=triangle*,mark options={solid,fill=red!30!white},forget plot]
%table[x=coverage ,y=gammaMax ,col sep=semicolon]
%{data_Gamma_cell_Nr4_050mW_new2.txt};
%	  
	  %%%%%%%%%
          
    \addplot[line width=1pt,dashed, color=OliveGreen] coordinates 
        {(500,0) (1400,0)};
    \addplot[line width=1pt,dashed, color=OliveGreen] coordinates 
        {(500,1) (1400,1)};
   
    \end{axis}

\end{tikzpicture}
}
\caption{Maximal ratio energy-to-interference $\Gamma_{\max}$}
\label{fig:FFR}
\end{figure}
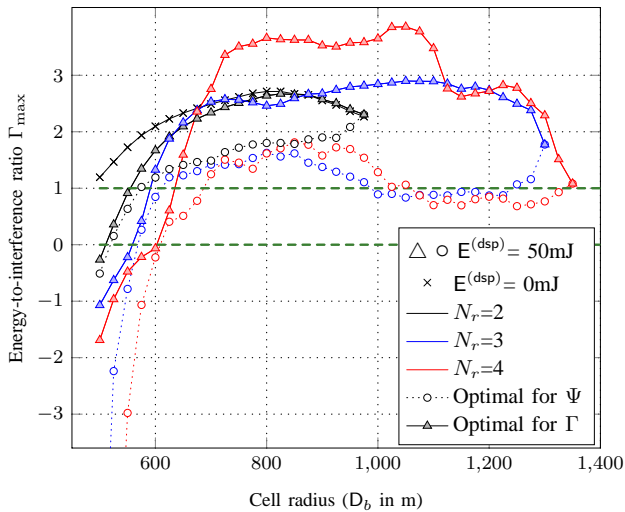

We now investigate the network performance in terms of both energy and interference, by using the $\Gamma$-metric. 
%%%%%%%%%%%%%%%%%%%%%%%%%%%%%%%%%%%%%%%%%%%%%%%%%%%%%
%\subsubsection{Impact of the relay deployment on $\Gamma$}
In Figure \ref{fig:FFR}, we plot the maximal feasible $\Gamma_{\max}$ as a function of the cell radius $\mathsf{D}_b$, for different values of $N_r$ and $\mathsf{E^{(dsp)}}$. To do so, for each considered considered set $\left( \mathsf{D}_b, N_r, \mathsf{E^{(dsp)}}\right)$, we find the location for the $N_r$ relays which maximizes $\Gamma$.

First, for a given number of relays, we observe that $\Gamma_{\max}$ is poor for small-size cells ($\mathsf{D}_b \leq$600m) and is significantly affected by $\mathsf{E^{(dsp)}}$. For example, at $\mathsf{D}_b \leq$550m, $\Gamma_{\max}$ is divided by 2 when $\mathsf{E^{(dsp)}}$ is increased from 0 to 50mW. On the contrary, for wider cells, the overall energy consumption is dominated by the RF transmit energy and the interference issue is relaxed due to distance. The impact of $\mathsf{E^{(dsp)}}$ is minor and $\Gamma_{\max}$ increases.
%When close to maximum cell radius,  $\Gamma_{\max}$ starts to decrease. Indeed, serving very far users implies high energy consumption at relay stations and thus, severe interference.
For the sake of clarity, the case $\mathsf{E^{(dsp)}}$=0mJ has been plotted for $N_r$=2 only, but results do not change for larger number of relays.

\begin{result}
Except for very short cell radius, the energy offset $\mathsf{E^{(dsp)}}$, consumed for decoding and re-encoding at the relay station, has little impact on the energy-to-interference ratio $\Gamma$.
\end{result}

Second, as depicted in Figure \ref{fig:FFR}, increasing the number of relays generally improves $\Gamma_{\max}$. We have shown in the previous subsection that relay configurations with $N_r$=3 and $N_r$=4 provide around the same minimal energy consumption per unit area $\Psi_{\min}$. On the contrary, when accounting for the interference generated by relays, the case $N_r$=4 largely outperforms $N_r$=3 and, for example, $\Gamma_{\max}$ is increased from 2.46 to 3.66 at $\mathsf{D}_b $=800m.
While the transmit energy gain achieved by increasing the number of relays from 3 to 4 is just enough to compensate for the additional idle energy $\mathsf{E_R^{(idle)}}$ (which is consumed whether or not data is transmitted and affects both $\upsilon_{\text{Gain}}$ and $\Psi_{\min}$), it is largely beneficial for the neighboring cells (reduced $\upsilon_{\text{Loss}}$). %, we have the following result:
\begin{result}
Accounting for the interference generated by relays, deploying many relay stations potentially closer to cell edge but transmitting at lower power is more efficient than deploying few relay stations far from cell edge but serving a large part of the cell.
\end{result}

%\subsubsection{A proposition for efficient relay deployment}

%An efficient relay deployment for cellular networks should balance the coverage, the energy consumption within the main cell and the interference generated in neighboring cells.
We now compare the performance achieved when the relay configuration is optimized either for $\Psi$ or $\Gamma$ and propose guidelines for efficient relay deployment.
To do so, Figure \ref{fig:TO_Energy_Coverage} also depicts the value for $\Psi$ obtained when the relay location is optimized for $\Gamma$ (namely, "Optimal for $\Gamma$" in the figure) and reversely, Figure \ref{fig:FFR} depicts the value for $\Gamma$ achieved by a location optimized for $\Psi$ ("Optimal for $\Psi$"). 
A first important remark is that both relay configurations are essentially distinct and provide notably different performance. To illustrate the gap between such deployment options, we plot in Figure \ref{fig:optimal_position} the relay configurations optimal for $\Gamma$ and for $\Psi$ with $N_r=2$ and various $\mathsf{D}_b$. 
Moreover, we observe from Figure \ref{fig:FFR} that the value for $\Gamma$ achieved with an energy-efficient relay deployment (optimal for $\Psi$) is below 1 for 1000m $\leq \mathsf{D}_b$ and $N_r \geq$ 3, meaning that the network performance is actually degraded.

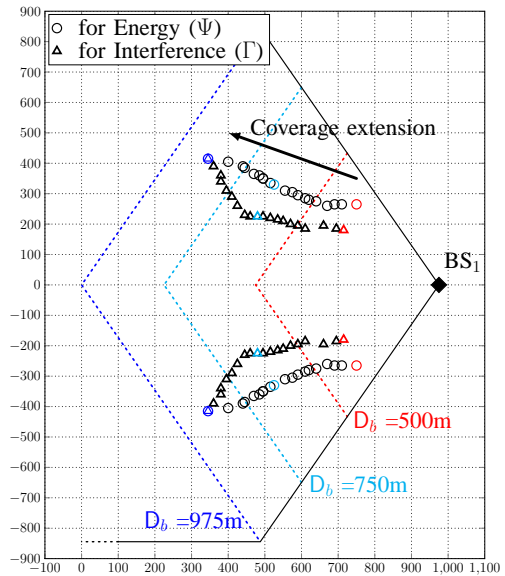
\begin{figure} 
\centering
\centering \resizebox{0.75\columnwidth}{!}{%

\begin{tikzpicture}[scale=4,cap=round,>=latex]

 \pgfplotsset{
    grid style = {
      dash pattern = on 0.05mm off 1mm,
      line cap = round,
      black,
      line width = 0.5pt
    }
  }

  \begin{axis}[%
%    xlabel=Blocks per kernel,%
%    ylabel=Info T/P (Mbps),%
	width=0.82\textwidth,
    height=\textwidth,
    	xmin=-100, xmax=1100, ymin=-900, ymax=900,
	%axis equal,
	every axis/.append style={font=\large},  
    grid=major,%
    legend style={at={(axis cs:-90,890)},anchor=north west, nodes=right, font=\huge},%
    %legend pos={south east},%
    mark size=2.0pt]

\addplot[only marks, color=black, mark size =4, mark=o, line width=1.0] table[x=xEEA ,y=yEEA ,col sep=semicolon] {data_Relay_position_Nr2_050mW_new.txt};
\addlegendentry{\; for Energy ($\Psi$)}  
\addplot[only marks, color=black, mark size =4, mark=triangle, line width=1.5] table[x=xGAMMA ,y=yGAMMA ,col sep=semicolon] {data_Relay_position_Nr2_050mW_new.txt};
\addlegendentry{\; for Interference ($\Gamma$)} 
%
%
% \addplot[only marks, color=black, mark size =5, mark=o, line width=1.2,forget plot] table[x=xMEAN ,y=yyMEAN ,col sep=semicolon] {data_Relay_position_Nr2_050mW_new.txt};
% 
 \addplot[only marks, color=black, mark size =4, mark=o, line width=1.0,forget plot] table[x=xEEA ,y=yyEEA ,col sep=semicolon] {data_Relay_position_Nr2_050mW_new.txt};
\addplot[only marks, color=black, mark size =4, mark=triangle, line width=1.5,forget plot] table[x=xGAMMA ,y=yyGAMMA ,col sep=semicolon] {data_Relay_position_Nr2_050mW_new.txt};

\addplot[only marks, color=blue, mark size =4, mark=o, line width=1.0]
    coordinates {    (345, 415)(345, -415)    };
    \addplot[only marks, color=blue, mark size =4, mark=triangle, line width=1.0]
    coordinates {    (345, 415)(345, -415)    }; 
\draw {(axis cs:487.5,  844.37) -- (axis cs:0,0) -- (axis cs:487.5,- 844.37)} [color=blue, line width=1.5pt, loosely dotted];
\draw (axis cs:450,-730) node[anchor = north east] {\color{blue} \huge $\mathsf{D}_b$ =975m};

\addplot[only marks, color=ProcessBlue, mark size =4, mark=o, line width=1.0]
    coordinates {    (525, 330)(525, -330)    };
    \addplot[only marks, color=ProcessBlue, mark size =4, mark=triangle, line width=1.5]
    coordinates {    (480, 225)(480, -225)    }; 
\draw {(axis cs:600,  649.52) -- (axis cs:225,0) -- (axis cs:600,-649.52)} [color=ProcessBlue, line width=1.5pt, loosely dotted];
\draw (axis cs:610,-660) node[anchor = west] {\color{ProcessBlue} \huge $\mathsf{D}_b$ =750m};

\addplot[only marks, color=red, mark size =4, mark=o, line width=1.0]
    coordinates {    (750, 265)(750, -265)
    };
    \addplot[only marks, color=red, mark size =4, mark=triangle, line width=1.5]
    coordinates {    (715, 180)(715, -180)
    };
\draw {(axis cs:725,  433.01) -- (axis cs:475,0) -- (axis cs:725,- 433.01)} [color=red,line width=1.5pt, loosely dotted];
\draw (axis cs:730,-450) node[anchor = west] {\color{red} \huge $\mathsf{D}_b$ =500m};

	\addplot[line width=1pt, color=black]
    coordinates {
    (100,  844.37)(487.5,  844.37)(975,0)
	(487.5,- 844.37)(100,- 844.37)
    };
            
     \draw {(axis cs:100,  844.37) -- (axis cs:0,  844.37)} [line width=1.25pt, loosely dotted];
     \draw {(axis cs:100,- 844.37) -- (axis cs:0,- 844.37)} [line width=1.25pt, loosely dotted];
     
	  \draw[->,line width=2.5pt, color=black] (axis cs:750,  350) -- (axis cs:400, 500);
	   
\coordinate (center) at (axis cs:0,0);
\coordinate (bs1) at (axis cs:975,0);
  
\draw (bs1) node[diamond, draw, fill = black,scale=1] {};
\draw (axis cs:1040,75) node {\huge BS$_1$};
\draw (axis cs:450,510) node [anchor=west] {\huge Coverage extension};
  
    \end{axis}

\end{tikzpicture}
}
\caption{Optimal relay positions ($N_r $ = 2, $\mathsf{E}^{(\Delta)}$ = 50mW)}
\label{fig:optimal_position}
\end{figure}

\begin{result}
Energy-efficient relay deployment does not necessarily lead to interference reduction and reversely, an interference-aware location is suboptimal for the cell energy consumption.%, due to the energy loss in neighboring cells.
\end{result}

%\textit{Remark:} We have performed additional simulations for a frequency reuse of 3. Here, $\Gamma_{\max}$ is very high (5-25) since the relay stations may only interfere with users located on a restricted network area, operating on the same frequency range. The optimal ratio $\Gamma_{\max}$ is dominated by the energy gain and the relay location optimal for $\Gamma$ is the same that the location optimal for the energy reduction, measured by $\Psi$.

Based on the above results, we propose a guideline for efficient relay deployment regarding both $\Psi$ and $\Gamma$. First, for short cell radius (550 $\leq \mathsf{D}_b \leq$ 700m), deploying two relay stations per sector located to minimize the energy consumption per unit area $\Psi$ can be considered as the best option, $\Gamma$ remains above 1, meaning that the overall network performance is not degraded.

Second, for wider cell size (700m $\leq \mathsf{D}_b$), deploying four relay stations per sector provides the optimal results for both $\Gamma$ and $\Psi$. However, current cellular networks are already reaching saturation and negotiating new site agreement for antenna deployment is getting ever harder for cellular operators. Thus, we argue that considering $N_r =$ 3 may actually be the best practical choice. When 700m $\leq \mathsf{D}_b \leq$ 1000m, optimizing the relay deployment for $\Gamma$ (resp. $\Psi$) does not degrade too much the performance in $\Psi$ (resp. $\Gamma$), such that both deployment options can be considered. However, for 700m $\leq \mathsf{D}_b$, a deployment optimized for $\Psi$ should not be considered since the overall network performance is degraded ($\Gamma < 1$).

\subsection{Impact of the relay coding scheme}

Up to now, we have shown that a relay deployment can be energy-efficient at the scale of a single cell (measured by $\Psi$) but without necessarily being efficient at a larger scale (measured by $\Gamma$). We now investigate the performance achieved by the energy-optimized relaying schemes described in Section \ref{sec:coding_scheme}. 
Regarding the $\Gamma$-metric, EO-PDF maximizes the energy gain $\upsilon_{\text{Gain}}$, while decreasing at the same time the energy loss $\upsilon_{\text{Loss}}$ experienced by neighboring cells. As detailed in \cite{Journal1}, the energy transmitted by EO-PDF is more uniformly spread over the two transmission phases and over both the direct and relaying links,  reducing at same time the power peaks causing high interference. The IR-PDF scheme minimized the use of the relay station by transmitting the most data possible via the direct link. The energy loss $\upsilon_{\text{Loss}}$ is minimized, but in return, the energy gain $\upsilon_{\text{Gain}}$ is reduced.
We recall that the circuitry consumption of such partial DF schemes is expressed as $\mathsf{E^{(dsp)}_{pdf}} = \mathsf{E_{2Hop}^{(dsp)}} + \mathsf{E_{pdf}^{(dsp+)}}$,
where $\mathsf{E_{pdf}^{(dsp+)}}$ is an additional offset accounting for their increased complexity.
%.  $\mathsf{E_{2Hop}^{(dsp)}}$ is the circuitry consumption of two-hop relaying (50mW for simulations) and $\mathsf{E_{pdf}^{(dsp+)}}$ is an additional offset accounting for the increased complexity of EO-PDF and IR-PDF.

\subsubsection{Objective and simulation settings}

For analysis, we consider a cell sector aided by two RS only ($N_r=2$), with low, medium and maximal cell radius ($\mathsf{D}_b \in \left\lbrace 600, 800, 975\text{m}\right\rbrace$).
Such configuration provides suboptimal performance in both $\Psi$ and $\Gamma$, compared to a configuration with more RS, but it offers valuable infrastructure cost reduction and deployment simplicity for a cellular operator. 
%Given results of Section \ref{result:guideline}, w
We consider as performance basis an energy-efficient relay deployment where both RS are located to minimize the energy per unit area $\Psi$ consumed by two-hop relaying.

To investigate how optimized relaying schemes can alleviate the interference issue, we derive the optimal utilization of coding schemes within the cell sector. To do so, we compare for each user location the performance achieved by two-hop relaying, EO-PDF and IR-PDF and select for each the one that increases $\Gamma$. Proceeding this way, we design a spatially-optimized utilization of coding schemes and draw a map showing the cell areas where to use each coding scheme to maximize $\Gamma$. Such map is illustrated in Figure \ref{fig:map_coding_scheme} for $\mathsf{D}_b$=800m and $\mathsf{E_{pdf}^{(dsp+)}} \in \left[0,50\right]$mJ.

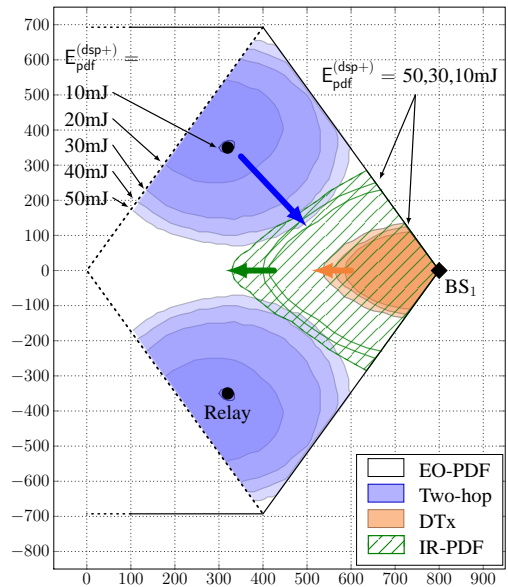
\begin{figure} 
\centering
\centering \resizebox{0.75\columnwidth}{!}{%

\begin{tikzpicture}[scale=4,cap=round,>=latex]

 \pgfplotsset{
    grid style = {
      dash pattern = on 0.05mm off 1mm,
      line cap = round,
      black,
      line width = 0.5pt
    }
  }

  \begin{axis}[%
%    xlabel=Blocks per kernel,%
%    ylabel=Info T/P (Mbps),%
	width=0.82\textwidth,
    height=\textwidth,
    	xmin=-75, xmax=950, ymin=-850, ymax=750,
	%axis equal,
	every axis/.append style={font=\Large},  
    grid=major,%
    legend style={nodes=right, font=\normalsize},%
    legend pos={south east},%
    mark size=2.0pt]

%% pour le fond
\draw {(axis cs:400,  692.82) -- (axis cs:0,0) -- (axis cs:400,-692.82) -- (axis cs:800,0)} [color=white, fill=white, line width=1pt];

%% repasser en blanc pour la couleur
\addplot[solid, color=white, fill=white, line width=1] table[x=x ,y=y ,col sep=semicolon] {data_EO_800_100mW.txt};    
\addplot[solid,  color=white, fill=white, line width=1] table[x=x ,y=yy ,col sep=semicolon] {data_EO_800_100mW.txt}; 
\addplot[solid, color=white, fill=white, line width=1] table[x=xcentre ,y=ycentre ,col sep=semicolon] {data_EO_800_100mW.txt};

%% EO: autour du relai
\addplot[solid, color=black, line width=1] table[x=x ,y=y ,col sep=semicolon] {data_EO_800_60mW.txt};
\addplot[solid, color=black, line width=1] table[x=x ,y=yy ,col sep=semicolon] {data_EO_800_60mW.txt};

\addplot[solid, color=black, fill=blue!60!white, opacity=0.25, line width=1] table[x=x ,y=y ,col sep=semicolon] {data_EO_800_70mW.txt};
\addplot[solid, color=black, fill=blue!60!white, opacity=0.25, line width=1] table[x=x ,y=yy ,col sep=semicolon] {data_EO_800_70mW.txt};

\addplot[solid, color=black, fill=blue!60!white, opacity=0.25, line width=1] table[x=x ,y=y ,col sep=semicolon] {data_EO_800_80mW.txt};
\addplot[solid,  color=black, fill=blue!60!white, opacity=0.25, line width=1] table[x=x ,y=yy ,col sep=semicolon] {data_EO_800_80mW.txt};

\addplot[solid, color=black, fill=blue!60!white,  opacity=0.25, line width=1] table[x=x ,y=y ,col sep=semicolon] {data_EO_800_90mW.txt};
\addplot[solid,  color=black, fill=blue!60!white,  opacity=0.25, line width=1] table[x=x ,y=yy ,col sep=semicolon] {data_EO_800_90mW.txt};

\addplot[solid, color=black, fill=blue!60!white,  opacity=0.25, line width=1] table[x=x ,y=y ,col sep=semicolon] {data_EO_800_100mW.txt};    
\addplot[solid,  color=black, fill=blue!60!white,  opacity=0.25, line width=1] table[x=x ,y=yy ,col sep=semicolon] {data_EO_800_100mW.txt};

%% au niveau du centre

\addplot[solid, color=black, fill=Orange,  opacity=0.35, line width=1] table[x=xcentre ,y=ycentre ,col sep=semicolon] {data_EO_800_60mW.txt}; 

\addplot[solid, color=black, fill=Orange, opacity=0.35, line width=1] table[x=xcentre ,y=ycentre ,col sep=semicolon] {data_EO_800_80mW.txt};

\addplot[solid, color=black, fill=Orange,  opacity=0.35, line width=1] table[x=xcentre ,y=ycentre ,col sep=semicolon] {data_EO_800_100mW.txt};

%% pour IR
    
\addplot[loosely dashed, color=green!50!black, pattern=my north east lines,line space=5pt, pattern color=green!50!black,  opacity=0.75, line width=0.5] table[x=x ,y=y ,col sep=semicolon] {data_IR_800_100mW.txt};   

\addplot[solid, color=green!50!black,  opacity=0.75, line width=1] table[x=x ,y=y ,col sep=semicolon] {data_IR_800_100mW.txt}; 

\addplot[solid, color=green!50!black,  opacity=0.75, line width=1] table[x=x ,y=y ,col sep=semicolon] {data_IR_800_80mW.txt};  
\addplot[solid, color=green!50!black,  opacity=0.75, line width=1] table[x=x ,y=yy ,col sep=semicolon] {data_IR_800_80mW.txt};        
    
\addplot[solid, color=green!50!black,  opacity=0.75, line width=1] table[x=x ,y=y ,col sep=semicolon] {data_IR_800_60mW.txt};  
\addplot[solid, color=green!50!black,  opacity=0.75, line width=1] table[x=x ,y=yy ,col sep=semicolon] {data_IR_800_60mW.txt};     
    
%% sector
    
\draw {(axis cs:400,  692.82) -- (axis cs:0,0) -- (axis cs:400,-692.82)} [color=black, line width=1.5pt, loosely dotted];

	\addplot[line width=1pt, color=black]
    coordinates {
    (100, 692.82)(400,  692.82)(800,0)
	(400,-692.82)(100,-692.82)
    };
            
     \draw {(axis cs:100,  692.82) -- (axis cs:0,  692.82)} [line width=1.25pt, loosely dotted];
     \draw {(axis cs:100,- 692.82) -- (axis cs:0,- 692.82)} [line width=1.25pt, loosely dotted];
     
	%  \draw[->,line width=2.5pt, color=black] (axis cs:750,  350) -- (axis cs:400, 500);
	   
\coordinate (center) at (axis cs:0,0);
\coordinate (bs1) at (axis cs:800,0);

\draw (bs1) node[diamond, draw, fill = black,scale=1] {};
\draw (axis cs:850,-50) node {\LARGE BS$_1$};
\draw (axis cs:320,-350) node[circle, draw, fill = black,scale=1] {};
\draw (axis cs:320,350) node[circle, draw, fill = black,scale=1] {};
\draw (axis cs:320,-410) node {\LARGE Relay};
 
%% Pour les flèches
\draw [->,color=green!50!black, line width=5pt](axis cs:425,0) -- (axis cs:315,0);
\draw [->,color=Orange, line width=5pt](axis cs:600,0) -- (axis cs:510,0);
\draw [->,color=blue, line width=5pt](axis cs:350,325) -- (axis cs:500,125); 
 
%% pour les annotations des energies
\draw (axis cs:550,500) node [anchor=south west] (text1) {\hspace*{-10pt}\LARGE $\mathsf{E^{(dsp+)}_{pdf}}=$ 50,30,10mJ};
\draw [->,color=black, line width=1pt](text1.south) -- (axis cs:665,260);
\draw [->,color=black, line width=1pt](text1.south) -- (axis cs:730,130);

\draw (axis cs:-60,550) node [anchor=south west] {\LARGE $\mathsf{E^{(dsp+)}_{pdf}}=$};
\draw (axis cs:-60,475) node [anchor=south west] (text10) {\LARGE 10mJ};
\draw [->,color=black, line width=1pt](text10.east) -- (axis cs:295,350);
\draw (axis cs:-60,400) node [anchor=south west] (text20) {\LARGE 20mJ};
\draw [->,color=black, line width=1pt](text20.east) -- (axis cs:168,310);
\draw (axis cs:-60,325) node [anchor=south west] (text30) {\LARGE 30mJ};
\draw [->,color=black, line width=1pt](text30.east) -- (axis cs:130,232);
\draw (axis cs:-60,250) node [anchor=south west] (text40) {\LARGE 40mJ};
\draw [->,color=black, line width=1pt](text40.east) -- (axis cs:105,205);
\draw (axis cs:-60,175) node [anchor=south west] (text50) {\LARGE 50mJ};
\draw [->,color=black, line width=1pt](text50.east) -- (axis cs:95,180);
 
%% Pour la légende 
% \fcolorbox{couleur cadre}{couleur fond}{texte}

 \node [rectangle, draw=black, fill=white, anchor= south east] (mylegend) at (axis cs:945,-840) { \LARGE 
\begin{tabular}{l l}
\tikz{ \draw [fill=white,  draw=black](0,0) rectangle (30pt,15pt); }  & EO-PDF \\
\tikz{ \draw [fill=blue!30!white,  draw=blue!80!black](0,0) rectangle (30pt,15pt); }  & Two-hop \\
\tikz{ \draw [fill=Orange!50!white,  draw=Orange!60!black](0,0) rectangle (30pt,15pt); }  & DTx \\
\tikz{ \draw [loosely dashed, color=green!50!black, pattern=my north east lines,line space=5pt, pattern color=green!50!black, line width=0.5](0,0) rectangle (30pt,15pt);
\draw [solid, color=green!50!black, line width=1](0,0) rectangle (30pt,15pt); } & IR-PDF 
%\\ \tikz{ \draw [->,color=black, line width=5pt](0,3pt) -- (30pt,3pt); } & Rising $\mathsf{E^{(PDF)}}$
\end{tabular}
};

    \end{axis}

\end{tikzpicture}
}
%
%%% Pour la légende 
%\resizebox{0.25\columnwidth}{!}{
%\fbox{ \Large
%\begin{tabular}{l l}
%\tikz{ \draw [fill=white,  draw=black](0,0) rectangle (50pt,25pt); } & EO-PDF \\
%\tikz{ \draw [fill=blue!30!white,  draw=blue!80!black](0,0) rectangle (50pt,25pt); } & Two-hop \\
%\tikz{ \draw [fill=Orange!50!white,  draw=Orange!70!black](0,0) rectangle (50pt,25pt); } & DTx \\
%\tikz{ \draw [loosely dashed, color=green!50!black, pattern=my north east lines,line space=5pt, pattern color=green!50!black, line width=0.5](0,0) rectangle (50pt,25pt);
%\draw [solid, color=green!50!black, line width=1](0,0) rectangle (50pt,25pt); } & IR-PDF \\
%\tikz{ \draw [fill=white,  draw=black](0,0) rectangle (50pt,25pt); }& Increasing $\mathsf{E^{(PDF)}}$
%\end{tabular}
%%\tikz{ \draw [fill=white,  draw=black](0,0) rectangle (10pt,6pt); }
%%$\mathsf{E}^{(\Delta)}$= 0mW \quad
%%\tikz{ \draw[fill=red!30!white,  draw=red] (0,0) rectangle (10pt,6pt); }
%%$\mathsf{E}^{(\Delta)}$= 50mW \quad
%%\tikz{ \draw [fill=brown!30!white,  draw=brown](0,0) rectangle (10pt,6pt); }
%%$\mathsf{E}^{(\Delta)}$= 100mW \quad
%}

%}
\caption{Optimal utilization of coding schemes in a cell sector}
\label{fig:map_coding_scheme}
\end{figure}

\subsubsection{Spatial analysis}

With a medium cell radius ($\mathsf{D}_b$=800m), when the additional circuitry consumption $\mathsf{E_{pdf}^{(dsp+)}}$ is below 10mJ, EO-PDF outperforms two-hop relaying in almost the whole cell area, as illustrated in Figure \ref{fig:map_coding_scheme}.
%The energy gain $\upsilon_{\text{Gain}}$ is increased by 10\% and the energy loss $\upsilon_{\text{Loss}}$ is decreased by more than 40\%.
For higher values of $\mathsf{E_{pdf}^{(dsp+)}}$, EO-PDF does not provide sufficient reduction in the transmit energy around the RS and cannot compensate for the dissipated energy $\mathsf{E_{pdf}^{(dsp+)}}$. However, even with $\mathsf{E_{pdf}^{(dsp+)}}$=50mJ (i.e. EO-PDF consumed twice as much energy as two-hop relaying to process data), EO-PDF still outperforms two-hop relaying when the user-to-relay link is weaker. 
In larger cells ($\mathsf{D}_b$=975m), EO-PDF outperforms two-hop relaying for any user location and any value of $\mathsf{E_{pdf}^{(dsp+)}} \in \left[0,50\right]$mJ. The additional circuitry consumption $\mathsf{E_{pdf}^{(dsp+)}}$ has only marginal effect on $\Gamma$. 
In smaller cells ($\mathsf{D}_b$=600m), the overall energy consumption is dominated by the circuitry consumption. Except from the case $\mathsf{E^{(dsp)}_{pdf}}= \mathsf{E_{2Hop}^{(dsp)}}$, the EO-PDF scheme improves the cell performance only if used very far from RS.% When the offset $\mathsf{E_{pdf}^{(dsp+)}}$ is high, two-hop relaying outperforms EO-PDF in almost the whole cell area.

%\begin{result}
%In a cell with medium or large coverage, the energy-optimized partial decode-forward scheme (EO-PDF) allows significant performance enhancement for most users, especially for users located far from the relay station. 
%\end{result}

Also note that the IR-PDF scheme can only reach the same performance as other coding schemes but without outperforming them. The corresponding cell areas are plotted in green.% in Figure 

\subsubsection{Coding schemes and relay deployment}

%\begin{figure*} 
%\centering
%\subfigure[Minimum energy per unit area $\Psi_{\min}$]{
%\input{figure_psiCS.tex}
%\label{fig:psi_coding_scheme}
%}
%\subfigure[Maximal energy-to-interference ratio $\Gamma_{\max}$]{
%\input{figure_gainCS.tex}
%\label{fig:gain_coding_scheme}
%}
%\centering \resizebox{0.4\columnwidth}{!}{%
%\hspace{10pt}\raisebox{120pt}{
%\tikz{ \node [rectangle, draw=black, fill=white, anchor= north west] (mylegend) at (0,0) { \Large
%\begin{tabular}{l l}
%\tikz{ \draw [solid, color=blue, line width=1.5](0,0) -- (1,0) ; } & $\mathsf{D}_b$=600m \\
%\tikz{ \draw [solid, color=green!50!black, line width=1.5](0,0) -- (1,0) ; } & $\mathsf{D}_b$=800m \\
%\tikz{ \draw [solid, color=red, line width=1.5](0,0) -- (1,0) ; } & $\mathsf{D}_b$=975m \\ 
%%
%\tikz{ \draw [dashed, color=black, line width=1.5](0,0) -- (1,0) ; } & 2Hop - $N_r=2$ \\
%\tikz{ \draw [solid, color=black, line width=1.5](0,0) -- (1,0) ; } & 2Hop - $N_r=3$ \\
%\tikz{ \draw [solid, color=black, line width=1.5](0,0) -- (1,0) ;
%\draw [solid, color=black, fill=black, line width=1.5](0.5,0) circle (3pt) ;
%} & EO-PDF / 2Hop - $N_r=2$ 
%\end{tabular}
%};
%}}}
%\caption{Performance gains with optimal utilization of coding scheme}
%\end{figure*}

\begin{figure} 
\centering
\centering \resizebox{0.85\columnwidth}{!}{%

\begin{tikzpicture}[scale=4,cap=round,>=latex]

 \pgfplotsset{
    grid style = {
      dash pattern = on 0.05mm off 1mm,
      line cap = round,
      black,
      line width = 0.5pt
    }
  }

  \begin{axis}[%
	xmin=-5, xmax=55,
	ymin=0.99, ymax=5.01,
	xlabel= Additional circuitry consumption $\mathsf{E^{(dsp+)}_{pdf}}$ in mJ,%
    ylabel=Energy-to-interference ratio $\Gamma_{\max}$,%
    y label style={at={(axis description cs:0.1,0.5)},rotate=0,anchor=south},
	every axis/.append style={font=\footnotesize},  
    grid=major,%
    legend style={at={(axis cs:38,4.60)},anchor=north west, nodes=right},%
    %legend pos={outer south},%
    mark size=2pt]

% 
%% pour la Légende:
 
\addplot[solid, color=blue, line width=1.2] table[x=e_pdf ,y=Gamma_optim ,col sep=semicolon] {data_CS_gamma_600m.txt};   
\addlegendentry{$\mathsf{D}_b$=600m}; 	
\addplot[solid, color=green!50!black, line width=1.2] table[x=e_pdf ,y=Gamma_optim ,col sep=semicolon] {data_CS_gamma_800m.txt};
\addlegendentry{$\mathsf{D}_b$=800m}; 
\addplot[solid, color=red, line width=1.2] table[x=e_pdf ,y=Gamma_optim ,col sep=semicolon] {data_CS_gamma_975m.txt};
\addlegendentry{$\mathsf{D}_b$=975m};  
 
\addplot[dashed, color=black, line width=1.2] coordinates {(0,1.669404) (50,1.669404) } ;
\addlegendentry{2Hop - $N_r=2$} ; 
\addplot[solid, color=black, line width=1.2] coordinates {(0,1.326) (50,1.326)} ;
\addlegendentry{2Hop - $N_r=3$} ; 

\addplot[solid, color=black, mark=*, line width=1.2] table[x=e_pdf ,y=Gamma_optim ,col sep=semicolon] {data_CS_gamma_600m.txt};   
\addlegendentry{EO-PDF / 2Hop};

%% pour les courbes:

\draw [dashed, color=blue, line width=1.2](axis cs: 0,1.669404) -- (axis cs:50,1.669404) ;
\draw [solid, color=blue, line width=1.2](axis cs: 0,1.326) -- (axis cs:50,1.326) ;

\draw [dashed, color=green!50!black, line width=1.2](axis cs: 0,2.638590) -- (axis cs:50,2.638590) ;
\draw [solid, color=green!50!black, line width=1.2](axis cs: 0,2.459) -- (axis cs:50,2.459) ;

\draw [dashed, color=red, line width=1.2](axis cs: 0,2.305105) -- (axis cs:50,2.305105) ; 
\draw [solid, color=red, line width=1.2](axis cs: 0,2.81497) -- (axis cs:50,2.81497) ;

\addplot[solid, color=blue, mark=*, line width=1.2] table[x=e_pdf ,y=Gamma_optim ,col sep=semicolon] {data_CS_gamma_600m.txt};   	
\addplot[solid, color=green!50!black, mark=*, line width=1.2] table[x=e_pdf ,y=Gamma_optim ,col sep=semicolon] {data_CS_gamma_800m.txt};
%\addlegendentry{$\mathsf{D}_b$=800m}; 
\addplot[solid, color=red, mark= *, line width=1.2] table[x=e_pdf ,y=Gamma_optim ,col sep=semicolon] {data_CS_gamma_975m.txt};

    \end{axis}

\end{tikzpicture}
}
\caption{$\Gamma_{\max}$ with optimal utilization of coding schemes}
\label{fig:gain_coding_scheme}
\end{figure}
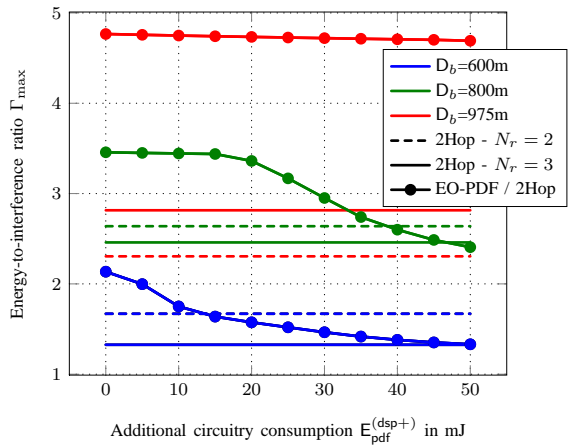

In the following, we denote "combination EO-PDF / 2Hop" as the spatially-optimized utilization of coding schemes, as previously described and illustrated in Figure \ref{fig:map_coding_scheme} and plot in Figure \ref{fig:gain_coding_scheme}, the maximal energy-to-interference ratio $\Gamma_{\max}$ achieved by this combination and by two-hop relaying only, with $N_r=2$ and $N_r=3$.

\textit{Remark:} For some data points, two-hop relaying only outperforms the combination EO-PDF / 2Hop.
Indeed, such combination is used with an energy-efficient deployment (optimized for $\Psi$), which is sub-optimal for $\Gamma$. On the contrary, the plotted performance of two-hop relaying is obtained with an interference-aware deployment (optimized for $\Gamma$) and can be understood as the maximum $\Gamma$ reached by two-hop relaying in the best feasible relay configuration. %We recall that a deployment optimized for $\Gamma$ shows in return suboptimal energy efficiency, as illustrated in Figure \ref{fig:TO_Energy_Coverage}.
%First, we focus on the case of a cell sector aided by two relays.
\begin{result}
In a small-size cell, the energy-optimized partial decode-forward scheme (EO-PDF) is severely affected by its increased circuitry consumption such that the optimal combination EO-PDF / 2Hop does not provide much performance enhancement when $\mathsf{E_{pdf}^{(dsp+)}} \geq $30mJ.
\end{result}
For example, when $\mathsf{D}_b$=600m and $\mathsf{E_{pdf}^{(dsp+)}} =$50mJ, the EO-PDF / 2Hop combination  achieves $\Gamma$=1.33, which outperforms two-hop relaying  with same relay location ($\Gamma$=1.19) but is below the performance of two-hop relaying when used with an interference-aware location ($\Gamma$=1.67).

\begin{result}
When the additional circuitry consumption $\mathsf{E_{pdf}^{(dsp+)}} $ is low, or when the cell radius is wide, the combination EO-PDF / 2Hop approaches and even outperforms the maximal $\Gamma_{\max}$ achieved by two-hop relaying, even with a relay location suboptimal for $\Gamma$. It also allows a reduction of the number of relays per sector for the same, or even better, ratio $\Gamma$.
\end{result}
As observed in Figure \ref{fig:gain_coding_scheme}, the maximal $\Gamma_{\max}$ achieved by the combination EO-PDF / 2Hop is higher than two-hop relaying with $N_r=3$, for almost any cell size and any additional circuitry consumption. It even outperforms the case $N_r=4$ for wide coverage extension ($\mathsf{D}_b$=975m).

\section{Conclusion}
\label{sec:conclusion}

We have highlighted a new trade-off on relay deployment for cellular networks that balances system energy efficiency and performance loss experienced by neighboring users due to the additional interference generated by  relays. To this end, we first formulated a spatial definition of the relay efficiency and proposed three tractable models allowing meaningful analysis without requiring time-consuming simulations. Next, we analyzed the correlative impact of the circuitry consumption, the location and number of relays as well as the relaying coding scheme on the network performance. By significantly reducing the transmit power peaks, energy-optimized coding schemes alleviate the interference issue, and by performing well even with suboptimal relay location or reduced number of relays, they offer valuable deployment flexibility.

\appendices 

\section{Proof of Lemma \ref{lemma:P_low}: Lower-bound for $\mathbb{P}_{\text{ER}}$}

\label{app:lower_bound}

The probability $\mathbb{P}_{\text{ER}}$ is expressed as follows:
\begin{align}
\mathbb{P}_{\text{ER}} = \overset{\mathsf{E_R^{(m)}}} {\underset{0}{\int}} \overset{\mathsf{E_B^{(m)}}} {\underset{0}{\int}}
\overset{\mathsf{E_B^{(m)}}} {\underset{0}{\int}}
\mathbb{P}\left( E_b + E_r \leq E_d \right)
d E_d d E_b d E_r
\label{eq:triple_integral}
\end{align}

Among the wide possibilities for lower bounds, we aim at removing the power constraints which condition $E_b + E_r \leq E_d$. Hence, we decompose $\mathbb{P}_{\text{ER}}$ into elementary probabilities that discard the triple integral. We get $\mathbb{P}_{\text{ER}} =  \mathbb{P}^{(1)} +\mathbb{P}^{(2)}$ with

\vspace*{-5pt}
{\small
 \begin{align*}
&\mathbb{P}^{(1)} = \mathbb{P}\left( E_b + E_r \leq E_d  \leq \mathsf{E_R^{(m)}} \right) 
\\
& \mathbb{P}^{(2)} = \mathbb{P} \left( \mathsf{E_R^{(m)}} \leq E_d  \leq \mathsf{E_B^{(m)}} \; \cap \; E_r \leq \mathsf{E_R^{(m)}} \; \cap \; E_b + E_r \leq E_d  \right).
\end{align*}
} \vspace*{-10pt}

First, we find a lower bound for $\mathbb{P}^{(1)} $. We have:

\vspace*{-5pt}
{\small
\begin{align*}
 \mathbb{P}& \left( E_b + E_r \leq E_d \right) = \;  \mathbb{P}\left( E_b + E_r \leq E_d  \leq \mathsf{E_R^{(m)}} \right)
 \\  
 + & \; \mathbb{P} \left( \mathsf{E_B^{(m)}} +\mathsf{E_R^{(m)}} \leq E_d \right)
 \mathbb{P} \left( E_b \leq \mathsf{E_B^{(m)}} \; \cap \;  E_r \leq \mathsf{E_R^{(m)}} \right)
\tag{a} \\  
 + & \; \mathbb{P} \left( \mathsf{E_B^{(m)}} +\mathsf{E_R^{(m)}} \leq E_d 
 \; \cap \; \left( \mathsf{E_B^{(m)}} \leq E_b  \; \cup \; \mathsf{E_R^{(m)}} \leq E_r \right)
 \right. \\ & \left.  
 \; \cap \; E_b + E_r \leq E_d
 \right)
\tag{b}\\ 
+ & \; \mathbb{P} \left( \mathsf{E_R^{(m)}} \leq E_d \leq \mathsf{E_B^{(m)}} +\mathsf{E_R^{(m)}} \right)
\mathbb{P} \left( E_b + E_r \leq \mathsf{E_R^{(m)}} \right)
\tag{c} \\ 
+ & \; \mathbb{P} \left( \mathsf{E_R^{(m)}} \leq E_d \leq \mathsf{E_B^{(m)}} +\mathsf{E_R^{(m)}}
\; \cap \; \mathsf{E_R^{(m)}} \leq E_b + E_r   \leq  E_d \;  \right)
\tag{d}
\end{align*}
} \vspace*{-5pt}

Regarding the probabilities of lines (a) and (c), the condition $E_b + E_r \leq  E_d$ necessarily holds given energy constraints and these two probabilities can be readily computed in closed-form.
Second, the probabilities of lines (b) and (d), denoted $\mathbb{P}^{(b)} $ and $\mathbb{P}^{(d)} $ respectively, can only be expressed in integral form, but are respectively upper-bounded by 
\begin{align*}
\mathbb{P}^{(b)}
 \leq \mathbb{P} &\left( \mathsf{E_B^{(m)}} + \mathsf{E_R^{(m)}} \leq E_d 
\, \cap \, \left( \mathsf{E_B^{(m)}} \leq E_s  \, \cup \, \mathsf{E_R^{(m)}} \leq E_r \right)
 \right)
\\
\mathbb{P}^{(d)}
\leq \mathbb{P}  & \left( \mathsf{E_R^{(m)}} \leq E_d \leq \mathsf{E_B^{(m)}} +\mathsf{E_R^{(m)}} \right.
\\ 
& \left. \cap \, \mathsf{E_R^{(m)}} \leq E_s + E_r    \leq  \mathsf{E_B^{(m)}} +\mathsf{E_R^{(m)}} \,  \right) 
\end{align*}
Plugging these upper-bounds into the expression for $ \mathbb{P} \left( E_b + E_r \leq E_d \right)$, we obtain the lower-bound $\mathbb{P}^{(1)}_{\text{low}} $ given in Eq. \eqref{eq:Prob_low1}.
Next, we have:
\begin{align*}
\mathbb{P}^{(2)} = & \;
\mathbb{P} \left( \mathsf{E_R^{(m)}} \leq E_d  \leq \mathsf{E_B^{(m)}} \; \cap \; E_b + E_r \leq \mathsf{E_R^{(m)}} \right)
\\ & +
\mathbb{P} \left( \mathsf{E_R^{(m)}} \leq E_b + E_r \leq E_d  \leq \mathsf{E_B^{(m)}} \; \cap \; E_r \leq \mathsf{E_R^{(m)}}\right)
\end{align*}
Note that the energy $E_r$ consumed by BS to transmit data to RS is generally low thanks to strong channel conditions. Thereby, the probability in second line approaches 0 and $\mathbb{P}^{(2)}$ can be tightly lower-bounded by $\mathbb{P}^{(2)}_{\text{low}} $ of Eq. \eqref{eq:Prob_low2}.
%Note that the proposed lower bound can be extended to other shadowing distributions than log-normal.
We now show that $\mathbb{P}_{\text{low}} = \mathbb{P}^{(1)}_{\text{low}} + \mathbb{P}^{(2)}_{\text{low}}$ can be approximated by a closed-form expression. First, 
$\mathbb{P}_{\text{low}}$, and thus $\mathbb{P}_{\text{CR}}$ and $\mathbb{P}_{\text{CD}}$, are computed using the cumulative distribution function $\Phi$ of the standard normal distribution:
\begin{align}
\mathbb{P} \left( E_k \leq \mathsf{E_k^{(m)}} \right)
= & \; \Phi \left(\frac{\ln \left(\mathsf{E_k^{(m)}} \right) - \mu_k}{\sigma_k} \right) 
\label{eq:erf}
\end{align}
where $\mu_k$ and $\sigma_k$ are given in Eq. \eqref{eq:E_i_0}. Even if $\Phi$ is written as an integral, it is widely available in scientific tools through well-known tables, such that its computation does not imply much complexity and can be considered as closed-form.
In addition, $\mathbb{P}_{\text{low}}$ requires the computation of $\mathbb{P} \left( E_b + E_r \leq \mathsf{E_B^{(m)}}+\mathsf{E_R^{(m)}}  \right)$ and
$\mathbb{P} \left( E_b + E_r \leq E_d  \right) = \mathbb{P} \left( E_{b}^{(0)} \frac{s_d}{s_b} + E_{r}^{(0)} \frac{s_d}{s_r} \leq  E_{d}^{(0)} \right)$,
both of which involve the sum of two log-normal random variables.
Such distributions do not have a closed-form expression, but have been extensively explored in the literature \cite{gao2009,beaulieu}. In this work, we consider the Fenton-Wilkinson approach and approximate these sum distributions by log-normal random variables.
$E_b + E_r$ is approximated by $E_{b+r} \sim \log \mathcal{N} \left(\mu_{b+r}, \sigma_{b+r}^2\right)$, where $\mu_{b+r}$ and $\sigma_{b+r}^2$ are computed as given in \cite[Eq. (9-12)]{beaulieu}.
Similar computation can be performed for $E_{b}^{(0)} \frac{s_d}{s_b} + E_{r}^{(0)} \frac{s_d}{s_r}$, taking into account the correlation coefficient between $E_{b}^{(0)} \frac{s_d}{s_b}$ and $E_{r}^{(0)} \frac{s_d}{s_r}$. 

Consequently, using the Fenton-Wilkinson approach, we have decomposed $\mathbb{P}_{\text{low}}(x,y)$ into elementary probabilities that can be computed in closed-form using Eq. \eqref{eq:erf}, thus avoiding the computation of a triple integral for each possible user location $M(x,y)$ within the cell.

\section{Proof of Lemma \ref{lemma:E_ER}: Bound for $\mathbb{P}_{\text{ER}} \mathbb{E} \left[E_{b+r} \, \vert \, \mathcal{C}_{\text{ER}}\right] $}

\label{app:E_ER}

Here, we use the decomposition for $\mathbb{P}_{\text{low}}$ that has been proposed in Lemma \ref{lemma:P_low}:
\begin{align*}
& \mathbb{P}_{\text{ER}} \mathbb{E} \left[E_{b+r} \, \vert \, \mathcal{C}_{\text{ER}}\right]
\leq \; E_1 + E_2 
\\ & \text{with} \quad
 E_1 = \mathbb{P}^{(1)} \mathbb{E} \left[E_{b+r} \, \vert \, E_b + E_r \leq E_d \leq \mathsf{E_R^{(m)}} \right] \quad \text{and} 
\\ 
&  E_2 = \mathbb{P}^{(2)}  \mathbb{E} \left[E_{b+r} \, \vert \, \mathsf{E_R^{(m)}} \leq E_d  \leq \mathsf{E_B^{(m)}} \; \cap \; E_b + E_r \leq \mathsf{E_R^{(m)}} \right]  \\
 & = \; g(b+r,\mathsf{E_R^{(m)}}) \mathbb{P}^{(2)}
\; \geq \; g(b+r,\mathsf{E_R^{(m)}}) \mathbb{P}_{\text{low}}^{(2)}
\end{align*} 
where $g$ is given by Eq. \eqref{eq:g}. With regards to $E_1$, we come back the integral form. Denoting $f_k$ the probability density function of $E_k$, we get:

{\small
\begin{align*}
& \hspace*{-15pt} E_1 =
\overset{\mathsf{E_R^{(m)}}} {\underset{0}{\int}}
\overset{\mathsf{E_d}} {\underset{0}{\int}}
\left(E_s + E_r\right) f_{s+r}\left( E_s + E_r \right) f_d\left(E_d \right) d(Es + Er) d E_d
\\
& \hspace*{-15pt} 
\simeq 
\overset{\mathsf{E_R^{(m)}}} {\underset{0}{\int}}
\exp \left( \mu_{s+r} + \frac{\sigma_{s+r}^2}{2} \right)
\Phi \left(- \sigma_{s+r} + \frac{\ln \left(E_d \right) - \mu_{s+r}}{\sigma_{s+r}}  \right)
f_d\left(E_d \right) d E_d
\\
& \hspace*{-15pt} 
= \exp \left( \mu_{s+r} + \frac{\sigma_{s+r}^2}{2} \right)
\mathbb{P}\left( \exp \left( \sigma_{s+r}^2 \right)(E_b + E_r) \leq E_d  \leq \mathsf{E_R^{(m)}} \right)
\end{align*}
}
%This completes the proof Lemma \ref{lemma:E_ER}.

\bibliographystyle{IEEEtranN}
{\footnotesize 
\bibliography{RefJournal3}
}
\fontsize{10}{10}
\selectfont

\end{document}